\documentclass[reqno]{amsart}
\usepackage[foot]{amsaddr}
\usepackage[margin=4cm]{geometry}
\usepackage{amsmath, amsfonts, amsthm, amssymb, mathtools, physics}
\usepackage{mathrsfs}
\usepackage{subcaption}
\usepackage[colorlinks = false]{hyperref}
\usepackage{graphicx}
\usepackage{algorithmic, algorithm}
\usepackage{booktabs}
\usepackage{enumerate}
\usepackage{bbm}
\numberwithin{equation}{section}
\graphicspath{{Figures/}}

\usepackage[textsize=tiny]{todonotes}

\mathtoolsset{showonlyrefs=true}
\newtheorem{theorem}{Theorem}[section]

\newtheorem{lemma}[theorem]{Lemma}
\newtheorem{proposition}[theorem]{Proposition}

\theoremstyle{definition}
\newtheorem{definition}[theorem]{Definition}
\newtheorem{remark}[theorem]{Remark}

\newtheorem{algo}[theorem]{Algorithm}
\newtheorem{notations}[theorem]{Notations}

\newcommand{\ind}{1\hspace{-2.1mm}{1}} \newcommand{\CC}{\mathbb{C}}
\newcommand{\RR}{\mathbb{R}}
\newcommand{\NN}{\mathbb{N}}
\newcommand{\PP}{\mathbb{P}}
\newcommand{\DD}{\mathbb{D}}

\newcommand{\EE}{\mathbb{E}}
\newcommand{\VV}{\mathbb{V}}
\newcommand{\FF}{\mathbb{F}}

\newcommand{\LL}{\mathbb{L}}

\newcommand{\Cc}{\mathcal{C}}
\newcommand{\Ee}{\mathcal{E}}
\newcommand{\Eeb}{\Ee^\blacklozenge}
\newcommand{\Call}{\mathrm{C}}
\newcommand{\Ff}{\mathcal{F}}
\newcommand{\Ss}{\mathcal{S}}

\newcommand{\Nn}{\mathcal{N}}
\newcommand{\Oo}{\mathcal{O}}

\newcommand{\Ii}{\mathcal{I}}

\newcommand{\Tt}{\mathcal{T}}

\newcommand{\Vv}{\mathcal{V}}

\newcommand{\M}{\mathrm{M}}

\newcommand{\D}{\mathrm{d}}
\newcommand{\E}{\mathrm{e}}
\newcommand{\I}{\mathrm{i}}
\newcommand{\Fm}{\mathrm{F}}
\newcommand{\Gm}{\mathrm{G}}

\newcommand{\BS}{\mathrm{BS}}

\newcommand{\Ffr}{\mathfrak{F}}

\newcommand{\ff}{\mathfrak{f}}

\newcommand{\cff}{\mathfrak{b}}
\newcommand{\cf}{\mathfrak{c}}

\newcommand{\Hp}{H_{+}}
\newcommand{\Hm}{H_{-}}
\newcommand{\Bh}{B^{H}}
\newcommand{\Ba}{B^{\frac{\alpha}{2}}}

\newcommand{\Dr}{\mathrm{D}}
\newcommand{\FfB}{\Ff^{\mathscr{B}}}
\newcommand{\FFB}{\FF^{\mathscr{B}}}
\newcommand{\Ef}{\Ee}
\newcommand{\Mf}{\mathcal{M}}

\newcommand{\half}{\frac{1}{2}}
\newcommand{\uu}{\boldsymbol{\mathrm{u}}}

\DeclareMathOperator*{\esssup}{ess\,sup}

\DeclareMathOperator{\VIX}{\mathrm{VIX}}

\def\blue#1{\textcolor{black}{#1}}

\title{Rough Bergomi turns grey}
\author{Antoine Jacquier}
\email{a.jacquier@imperial.ac.uk}
\address{Department of Mathematics, Imperial College London}
\author{Adriano Oliveri Orioles}
\email{adrian.orioles@gmail.com}
\address{Department of Mathematics, Imperial College London}
\author{Žan Žurič}
\email{zan.zuric@gmail.com}
\address{}
\date{\today}
\keywords{Rough volatility, multi-factor, asymptotics, VIX, Malliavin calculus}
\subjclass{60G15, 60G22, 60H07, 91G20}
\thanks{AJ acknowledges financial support from the EPSRC grant EP/T032146/1.}

\begin{document}
\begin{abstract}
We propose a tractable extension of the rough Bergomi model, 
replacing the fractional Brownian motion with a generalised grey Brownian motion, which we show to be reminiscent of models with stochastic volatility of volatility.
This extension breaks away from the log-Normal assumption of rough Bergomi, thereby making it a viable suggestion for the Equity Holy Grail --- the joint SPX/VIX options calibration.
For this new (class of) model(s), we provide semi-closed and asymptotic formulae for SPX and VIX options and show numerically its potential advantages as well as calibration results.
\end{abstract}

\maketitle

\tableofcontents

\section{Introduction}

Fractional Brownian motion has a long and famous history in probability, stochastic analysis, Physics and their applications to diverse fields~\cite{hurst1951long, kolmogorov1940wienersche, mandelbrot1968fractional}.
It was recently rejuvenated in the form of fractional volatility models in mathematical finance. 
First introduced by Comte and Renault~\cite{comte1996long}, 
later studied theoretically by Djehiche and Eddahbi~\cite{djehiche2001hedging}, 
Al\`os, Le\'on and Vives~\cite{alos2007short} and Fukasawa~\cite{fukasawa2011asymptotic}, they were properly empirically justified and promoted by Gatheral, Jaisson and Rosenbaum~\cite{Gatheral2018VolatilityRough} 
and Bayer, Friz and Gatheral~\cite{Bayer2015PricingVolatility}. 
Since then, a vast literature has pushed the analysis in many directions~\cite{jaber2022characteristicfunctiongaussianstochastic, jaber2024statespacesmultifactorapproximations, bank2023rough, bayer2022markovian, bayer2023rough,Bourgey2024, gatheral2020quadraticroughhestonmodel, Jacquier2021RoughOptions},
leading to theoretical and practical challenges (and exciting progress) to understand and implement these models.

While these rough stochastic volatility models may ultimately not be the universally perfect models they may have seemed at first --- and new ones in the form of multifactor or path-dependent models are fascinating competitors --- they nevertheless helped re-design a new framework for volatility modelling, 
spearheading new advances in the joint calibration of SPX and VIX options, the so-called \emph{Holy Grail} of Equity modelling.
Among these rough volatility models, the rough Bergomi~\cite{Bayer2015PricingVolatility} model, an extension of the classical Bergomi model~\cite{Bergomi}, was first highlighted for its tractability.
Indeed, as with its classical cousin, the variance process is log-Normal (though non-Markovian) and thus more readily amenable to computations.
Unfortunately, this property also implies that the squared VIX, as a conditional expectation of the  integrated variance, is also close to log-Normal, implying an almost flat VIX volatility smile, 
markedly different from the upward curve observed on the market.
One may, of course — and this has indeed been done — construct more sophisticated versions, such as rough Heston~\cite{euch2018perfect, el2019characteristic}, though at the cost of losing tractability (even though the simplified version~\cite{gatheral2020quadraticroughhestonmodel} appears very promising).

With this in mind, we introduce an extension of the rough Bergomi (rBergomi) model, which we call the Grey Bergomi (gBergomi) model:
it preserves self-similarity and stationary increments while relaxing the log-Normality constraint. 
This is inspired by and reminiscent of (rough)  stochastic volatility of volatility models or volatility-modulated models, as studied in~\cite{barndorff2012stochastic, Fouque2018HestonOptions, horvath2020volatility, huang2019volatility}.
We develop the mathematical background, the model and its pricing characteristics in Section~\ref{sec:ggBm_definitions}
and gather the required numerical algorithms in Section~\ref{sec:algos}.
Since some pricing formulae may be complicated and cumbersome, we develop some asymptotic closed-form approximations in Section~\ref{sec:asymptotics} 
and gather all these tools in Section~\ref{sec:calibration} to examine the calibration properties of this model on SPX and VIX options data.

%%%%%%%%%%%%%%%%%%%%%%%%%%%%%%%%%%%%
%%%%%%%%%%%%%%%%%%%%%%%%%%%%%%%%%%%%
\section{The grey Bergomi model}\label{sec:ggBm_definitions}

We start in Section~\ref{sec:ggBm} by introducing the so-called generalised grey Brownian and its properties.
This allows us, in Section~\ref{sec:gBergomi}, to extend the rough Bergomi model to its generalised counterpart 
\blue{(and provide a comparison in Remark~\ref{rem:Comparison})}, 
highlighting some of its useful properties, 
before showing in Section~\ref{sec:gBergomiVIX} how the VIX and VIX Futures take shape in this model.

\subsection{Generalised grey Brownian motion}
\label{sec:ggBm}

For $\beta>0$, the standard Mittag-Leffler function~$\Ef_{\beta}$ is defined as an entire function via the series representation
\begin{equation}\label{def:MLf}
\Ef_{\beta}(z):=\sum_{n\geq 0} \frac{z^n}{\Gamma(\beta n+1)},
\qquad\text{for all }z \in \CC,
\end{equation}
where~$\Gamma$ is the Gamma function and, for $\beta\in(0, 1]$,
the $\M$-Wright function~$\Mf_{\beta}$ reads
\begin{equation}\label{eq:M_Wright}
\Mf_{\beta}(z) := 
\sum_{n\geq 0} \frac{(-z)^n}{n ! \Gamma(-\beta n+1-\beta)}, \quad\text{for }z \in \CC.
\end{equation}

\blue{
\begin{remark}
    Recall that the Gamma function is defined for negative values via its analytic continuation 
    $$
    \Gamma(-x) = \frac{1}{\Gamma(x+1)}\frac{\pi}{\sin(\pi(x+1))},
    $$
    where $x\in\RR$.
\end{remark}
}
The choice $\beta=\half$ reduces the $\M$-Wright function to the Gaussian density, 
as a simple computation shows that 
$\Mf_{\half}(z)
  = \frac{1}{\sqrt{\pi}}\exp\left\{-\frac{z^2}{4}\right\}$.
The two functions~$\Mf_{\beta}$ and~$\Ef_{\beta}$ are related through the Laplace transform 
\begin{equation}\label{eq:LaplaceM-Wright}
\int_0^{\infty} \E^{-s u} \Mf_{\beta}(u) \D u
= \Ef_{\beta}(-s),
\qquad\text{for any }s\in \RR.
\end{equation}
Finally, a random variable~$Y_\beta$ is said to follow the (one-sided) $\M$-Wright distribution
if it is supported on the positive half line and admits 
the~$\M$-Wright function~\eqref{eq:M_Wright}
as probability density.
Note in particular that its moment of order $\kappa>-1$ exists~\cite{Piryatinska2005ModelsCase} and 
\begin{equation}\label{eq:MWrightMom}
\EE \left[Y_\beta^\kappa\right]
 = \frac{\Gamma(1+\kappa)}{\Gamma(1+\beta \kappa)}.
\end{equation}

We are now ready to introduce the generalised grey Brownian motion.

\begin{definition}\label{def:ggBm}
Let $\beta\in(0,1]$, $\alpha\in(0,2)$. 
A generalised grey Brownian motion $B^{\beta, \alpha}$ defined on a complete probability space $(\Omega, \Ff, \PP)$ is a one-dimensional continuous process  starting from 
$B^{\beta, \alpha}_0 = 0$ $\PP$-almost surely,
such that, for any $0 \leq t_1<\ldots<t_n<\infty$, its joint characteristic function is given by
$$
    \EE\left[\exp\left\{\I \sum_{k=1}^n u_k B^{\beta, \alpha}_{t_k}\right\}\right]
    = \Ef_{\beta}\left(-\half\uu^\top\Sigma_\alpha \uu\right),
    \quad \text{ for any }\uu=(u_1, \ldots, u_n) \in\RR^n,
$$
\blue{
where
$\displaystyle \Sigma_\alpha := \half\left(t_k^\alpha+t_j^\alpha-\left|t_k-t_j\right|^\alpha\right)_{k, j=1}^n$.
}
\end{definition}

\blue{
\begin{remark}
The existence of the distribution of the generalised grey Brownian motion~$B^{\beta,\alpha}$ is a consequence of Minlos Theorem \cite[Theorem~1.1 and Equation~(37)]{mura2007classselfsimilarstochasticprocesses}
\end{remark}
}
Note that~$\Sigma_\alpha$ is the covariance matrix of a (Mandelbrot-Van Ness, or Type-I) fractional Brownian motion with Hurst exponent~$\frac{\alpha}{2}$ and is symmetric positive definite.
\begin{remark}
When $\beta=\alpha=1$, the generalised grey Brownian motion reduces to a standard Brownian motion, as one can indeed check that
$$
\Ef_{1}\left(-\half\uu^\top\Sigma_1 \uu\right) = \exp\left\{- \half\uu^\top \left(\min\{t_i, t_j\}\right)_{i,j=1}^n \uu\right\}.
$$
In the Physics literature pertaining to anomalous diffusions, 
the~$\half$ factor is not present in~$\Sigma_{\alpha}$,
as Physics conventions normalise Brownian motion to a variance of~$2$ at time~$1$~\cite[Chapter~10]{ibe2013markov} as opposed to a normalisation of~$1$, 
standard in probability theory ~\cite{mura2007classselfsimilarstochasticprocesses}. 
\end{remark}
By inverse Fourier transform, the joint characteristic function above is integrable and decays rapidly; therefore the distribution is absolutely continuous and the joint density of $(B^{\beta, \alpha}_{t_1}, \ldots, B^{\beta, \alpha}_{t_n})$ reads~\cite{daSilva2020SingularityMotion}
$$
    f_\beta(\uu) = 
    \frac{(2\pi)^{-\frac{n}{2}}}{\sqrt{\operatorname{det} \Sigma_\alpha}} 
    \int_0^{\infty}
    \tau^{-\frac{n}{2}}\exp\left\{-\frac{\uu^{\top} \Sigma_\alpha^{-1} \uu}{2\tau}
    \right\} \Mf_{\beta}(\tau) \D \tau,
    \qquad\text{for all }\uu\in\RR^{n}.
$$
By simple computations, 
for any $s,t \geq 0, n \in\NN$,
odd moments of~$B_t^{\beta, \alpha}$ are null and
$$
\EE\left[\left(B_t^{\beta, \alpha}\right)^{2 n}\right]
= \frac{(2n)! t^{n\alpha}}{2^n \Gamma(\beta n+1)},
\qquad\text{and}\qquad
\EE\left[B_t^{\beta, \alpha} B_s^{\beta, \alpha}\right]
= \frac{t^\alpha+s^\alpha-|t-s|^\alpha}{2 \Gamma(\beta+1)}.
$$
Furthermore, for each $t, s \geq 0$, the characteristic function of the increments reads
\begin{equation}\label{eq:gBm_char0} \EE\left[\exp\left\{\I u\left(B_t^{\beta, \alpha}-B_s^{\beta, \alpha}\right)\right\}\right]
    = \varphi_{t-s}(u), 
    \qquad \text{for all } u \in \RR, 
\end{equation}
    where
\begin{equation}\label{eq:gBm_char}
\varphi_{\delta}(u):=\Ef_{\beta}\left(-\frac{u^2}{2}|\delta|^\alpha\right), \quad \text{for all } \delta, u\in\RR.
\end{equation}
Since the Mittag-Leffler function~$\Ef_{\beta}$ is not quadratic, 
the marginals of~$B^{\beta, \alpha}$ are not Gaussian,
and~\eqref{eq:gBm_char0} shows that it is $\frac{\alpha}{2}$-self-similar with stationary increments.
A careful analysis~\cite{daSilva2016SingularityParameters} further shows that the sample paths of~$B^{\beta, \alpha}$ have finite $p$-variation for any $p>\frac{2}{\alpha}$, thereby implying that it is not a semimartingale whenever $\alpha\in(0,1)$.
When $\alpha \in (1,2)$, an argument similar to~\cite{rogers1997arbitrage} shows that it cannot be a semimartingale either since its quadratic variation is null 
(which would imply finite variation, incompatible with its $1$-variation being infinite).
We highlight one interesting property that will be key for our computations and which will help provide financial intuition later on:

\begin{lemma}[Proposition~3~in~\cite{Mura2008}]\label{lem:gBm_decomposition}
\blue{The process $B^{\beta,\alpha}$ has the same finite-dimensional distributions as $\sqrt{Y_\beta}B_t^{\frac{\alpha}{2}}$, where~$\Ba$ is a Mandelbrot--Van Ness fractional Brownian motion with Hurst parameter~$\frac{\alpha}{2}$ and~$Y_\beta$ an independent one-sided $\M$-Wright random variable.}
\end{lemma}

%\adri{Change: addressed Hölder continuity via Kolmogorov's criterion rather than the (distributional) decomposition lemma, which the referee had suggested (Reviewer 2, Comment 5).}
\blue{
\begin{remark}
The decomposition in Lemma~\ref{lem:gBm_decomposition} holds in finite-dimensional
distribution and not pathwise, so the Hölder regularity of $B^{\alpha/2}$ does not transfer
directly to $B^{\beta,\alpha}$ along paths, but follows from Kolmogorov's continuity theorem. 
By stationarity of the increments and Lemma~\ref{lem:gBm_decomposition},
$B^{\beta,\alpha}_t - B^{\beta,\alpha}_s$ and $\sqrt{Y_\beta}B^{\alpha/2}_{t-s}$
are equal in distribution for each~$t$,
whence, for any integer $n\geq1$,
$$
\EE\big[|B^{\beta,\alpha}_t - B^{\beta,\alpha}_s|^{2n}\big]
= \EE\big[Y_\beta^{n}\big]\EE\big[|B^{\alpha/2}_{t-s}|^{2n}\big]
= \frac{(2n)!}{2^{n}\Gamma(\beta n + 1)}|t-s|^{n\alpha},
$$
using $\EE[Y_\beta^n]=\Gamma(n+1)/\Gamma(\beta n+1)$ from~\eqref{eq:MWrightMom} together with
the Gaussian moments of~$B^{\alpha/2}$ 
(with variance $|t-s|^{\alpha}$). 
Choosing $n>\frac{1}{\alpha}$,  Kolmogorov's continuity theorem provides a modification of~$B^{\beta,\alpha}$ that is almost surely locally H\"older continuous of any order $\gamma<\frac{n\alpha-1}{2n}$. 
Letting $n\uparrow\infty$,
$B^{\beta,\alpha}$ admits a modification that is locally Hölder continuous of any order
$\gamma<\frac{\alpha}{2}=H$, as for the fractional Brownian motion itself.
\end{remark}
}

Given the characteristic function~\eqref{eq:gBm_char0}-\eqref{eq:gBm_char}
for $0<s<t$, define its analytic extension
\begin{equation}\label{eq:ggBmMGF}
    \mathfrak{M}_{t-s}(u)
    \coloneqq \varphi_{t-s}\left(-\I u\right)
    = \EE\left[\E^{ u\left(B_t^{\beta, \alpha} - B_s^{\beta, \alpha}\right)}\right]
    = \Ef_{\beta}\left(\frac{u^2}{2}|t-s|^\alpha\right),
    \quad\text{for all }u\in\RR,
\end{equation}
well-defined and positive,
thus fully characterising the moment-generating function.

%%%%%%%%%%%%%%%%%%%%%%%%%%%%%%%%%%%%%%%%%%%
%%%%%%%%%%%%%%%%%%%%%%%%%%%%%%%%%%%%%%%%%%%
\subsection{The grey Bergomi model}\label{sec:gBergomi}
Our starting point is a simple variation of the rough Bergomi model for an underlying stock price~$S$,
originally proposed by Bayer, Friz and Gatheral~\cite{Bayer2015PricingVolatility}, with risk-neutral dynamics (assuming no interest rate)
\begin{equation}\label{eq:rBergomi_dynamics}
\begin{array}{r@{\;}lc}
\displaystyle \frac{\D S_t}{S_t} &= \displaystyle \sqrt{V_t} \D W_t, & S_0=1, \\
V_t &= \displaystyle \xi_0(t) \Ee^\lozenge(\eta \Gm_t), & V_0>0,
\end{array}
\end{equation}
where $\xi_0(\cdot)>0$ denotes the initial (forward) variance curve, $\eta>0$
is a (volatility of volatility) parameter, and~$\Gm$ a continuous non-degenerate Gaussian process.
Standard Gaussian theory ensures the existence of a standard Brownian motion~$B$,
correlated with~$W$ with correlation $\rho\in[-1,1]$,
such that
$\Gm_t = \int_{0}^{t}K(t-s)\D B_s$ for some kernel function~$K \in L^2((0,\infty))$.
%and both noise processes~$W$ and~$B$ are adapted to the same canonical filtration $\Ff^W \vee \Ff^{\Gm}$.
Here~$\Ee^\lozenge$ denotes the Wick exponential defined as~$\Ee^\lozenge(Z)\coloneqq \exp\{Z-\half\EE[Z^2]\}$
(note that the classical Dol{\'e}ans-Dade is not well defined for fractional Brownian motion as the latter is not a semimartingale~\cite{rogers1997arbitrage} for $H \ne\half$).
This version of the rough Bergomi model exhibits VIX marginals that closely resemble a log-Normal distribution, hence implying an almost flat VIX smile;
this was already noted in~\cite{Bayer2015PricingVolatility},
but is not consistent with the observed upward-sloping behaviour of the VIX smile, and we showcase some empirical evidence in Appendix~\ref{apx:VIX_smile_rBergomi}.
In a stochastic volatility framework, 
evidence has highlighted that 
the volatility of the volatility parameter should not be deterministic~\cite{barndorff2012stochastic, Fouque2018HestonOptions}, 
leading some to believe that a stochastic formulation could help reconcile the calibration of SPX and VIX smiles.
Furthermore, the inclusion of additional factors for an adequate joint fit is well-documented and supported by the literature:
in particular R{\o}mer~\cite{Rmer2022EmpiricalMarkets} investigated calibration problems of one-factor rough volatility models over 2004–2019 and concluded that the joint calibration problem is largely solvable with two-factor volatility models. 
This was further documented for rough volatility models~\cite{Jacquier2021RoughOptions}, 
regime-switching dynamics~\cite{Goutte2017Regime-switchingOptions}, and path-dependent volatility models~\cite{Guyon2022VolatilityPath-Dependent}.

In this spirit, we introduce a new, strictly positive, factor as a random volatility of volatility~$\widetilde{\eta}$, so that the rough Bergomi variance process in~\eqref{eq:rBergomi_dynamics} is replaced by
$$
V_t = \xi_0(t) \Ee^\blacklozenge\left(\widetilde{\eta} \Gm\right)_t, \qquad \xi_0(t)>0, \quad t\in[0,T],
$$
where, for any $\eta>0$,
\begin{equation}
    \Eeb(\eta Z)_{t}
    \coloneqq \exp\left\{\eta Z_t - \log\mathfrak{M}_{t}(\eta)\right\}, \quad \text{ for }  t\in[0,T],
\end{equation}
with~$\mathfrak{M}$ introduced in~\eqref{eq:ggBmMGF}.
Note that this collapses to~$\Ee^{\lozenge}(\eta Z)$ 
for~$Z$ centred Gaussian.
By definition, the forward variance reads $\xi_{s}(t) := \EE[V_t | \Ff_{s}^\Gm]$
for $0\leq s\leq t$, and
$\xi_{0}(t) = \EE[V_t]$ is
guaranteed as $\EE\left[\Eeb(\eta Z)\right] = 1$.
We may in particular choose $\widetilde{\eta} = \eta\sqrt{Y_\beta}$ with $\eta>0$ 
and~$Y_\beta$ an independent non-negative random variable with density~$\Mf_{\beta}$ in~\eqref{eq:M_Wright},
in which case, by Lemma~\ref{lem:gBm_decomposition}, 
this \emph{grey Bergomi} 
can be viewed as an extension of~\eqref{eq:rBergomi_dynamics}
with randomised volatility of volatility, where
\begin{equation}\label{eq:gBergomi_dynamics_old}
V_t = \displaystyle \xi_0(t) 
\Eeb\left(\eta B^{\beta, \alpha}_t\right),\qquad V_0 > 0.
\end{equation}
where~$B^{\beta, \alpha}$ is a generalised grey Brownian motion from Definition~\ref{def:ggBm}.
For computational reasons however, the fractional Brownian motion~$\Ba$ has some drawbacks, 
in particular the fact that its integral (Mandelbrot-van Ness) representation against a standard Brownian motion 
requires a rather cumbersome kernel.
We thus amend the above setup slightly (similar to~\cite{Bayer2015PricingVolatility}),
replacing~$\Ba$ by a Riemann-Liouville fractional Brownian motion. %\adri{Change: Don't be too quick and outline the loss of stationary increments (Reviewer 1, Comment 3)}
\blue{This change does not preserve the stationary increment structure outlined in Section~\ref{sec:ggBm}, and with the use of Lemma~\ref{lem:gBm_decomposition} we propose an alternative system to that of~\eqref{eq:gBergomi_dynamics_old}}
\begin{equation}\label{eq:gBergomi_dynamics}
\left\{
\begin{array}{rll}
\displaystyle \frac{\D S_t}{S_t} & 
= \displaystyle \sqrt{V_t}\left(\rho \D B_t + \sqrt{1-\rho^2}\D W_t\right),  & S_0 = 1,\\
V_t & = \displaystyle \xi_0(t) 
\Eeb\left(\eta \sqrt{Y_\beta}\Bh_t\right), & V_0 > 0.
\end{array}
\right.
\end{equation}
where from now on, $\Bh$ denotes a Riemann-Liouville fractional Brownian motion with Hurst parameter $H \in (0,1)$
(from the results in Section~\ref{sec:ggBm}, we have the correspondence $\alpha = 2H$), admitting the representation
$$
\Bh_t = \cf\int_{0}^{t}(t-s)^{\Hm}\D B_s,
\qquad\text{for all }t\geq 0,
$$
with the following convenient notations from now on:

\begin{notations}
$\displaystyle \cf := \frac{1}{\Gamma(\Hp)}$,
$H_{\pm} := H\pm\half$,
$\displaystyle \cff := \frac{\eta^2 \cf^2}{4H}$.
\end{notations}

\blue{
\begin{remark}
Explicitly, the renormalising factor in~\eqref{eq:gBergomi_dynamics} is the unconditional
moment generating function of the grey driver: since $\Bh_t$ (conditional on $\sigma(Y_{\beta})$) is centred Gaussian with variance $\cf^2 t^{2H}/(2H)$,
$$
\mathfrak{M}_t(\eta)
:= \EE\left[\exp\{\eta\sqrt{Y_\beta}\Bh_t\}\right]
= \EE\left[\exp\{\cff Y_\beta t^{2H}\}\right]
= \Ee_\beta\left(\cff t^{2H}\right),
$$
the last equality being the $\M$-Wright moment generating function~\eqref{eq:MWrightMom}. Hence the variance process~\eqref{eq:gBergomi_dynamics} admits the explicit form \begin{equation}\label{eq:variance_explicit}
V_t
= \xi_0(t)\Eeb\left(\eta\sqrt{Y_\beta}\Bh_t\right)
= \frac{\xi_0(t)}{\Ee_\beta\left(\cff t^{2H}\right)}
  \exp\left\{\eta\cf\sqrt{Y_\beta}\int_0^t (t-s)^{\Hm}\D B_s\right\}.
\end{equation} Here $\eta>0$ is the deterministic volatility-of-volatility parameter and the randomisation
enters solely through $\sqrt{Y_\beta}$; the composite $\widetilde{\eta}=\eta\sqrt{Y_\beta}$
introduced above is used only for motivation and not carried further. Because the
renormalisation $\Ee_\beta(\cff t^{2H})$ is deterministic, $\EE[V_t]=\xi_0(t)$ holds unconditionally.
\end{remark}
}

%\adri{Change: removed the incorrect equality $\sigma(Y_\beta B_t) = \sigma(Y_\beta) \lor \Ff_t^B$ (Reviewer 2, Comment 7)}
%\adri{Reviewer 2, Comment 11: yes, $\Ff^B_0 = \sigma(B_0)$ is the trivial filtration, since $B_0 = 0$ almost surely so $\sigma(B_0) = \{\emptyset, \Omega\}$. No change needed.}
\blue{In the sequel, we denote for any $t>0$ the sigma-algebras $\Ff^Z_t\coloneqq \sigma(Z_u, u\le t)$ for $Z\in\{W,B\}$, $\FfB_t\coloneqq \sigma(Y_\beta) \lor \Ff_t^B$ and $\Ff_t\coloneqq \Ff^W_t\lor \FfB_t$.
Finally, the filtrations are denoted by~$\FF^Z:=\cup_{t\geq 0}\Ff^Z_t$ for $Z\in\{W,B\}$
and $\FF \coloneqq \FF^W \lor\FFB$.}

\begin{remark}
We refer to~\eqref{eq:gBergomi_dynamics} as the gBergomi model despite the use of the Riemann-Liouville representation for the fractional Brownian motion.
\end{remark}
In order for this model to make sense, the Fundamental Theorem of Asset Pricing requires the stock price to be a true martingale:
\begin{theorem}\label{thm:bbBmartingale}
If $\rho \leq 0$, the stock price $(S_t)_{t\geq 0}$, solution to~\eqref{eq:gBergomi_dynamics}, is a true $\FF$-martingale. 
\end{theorem}
\begin{proof}
Following the ideas developed in~\cite{gassiat2019martingale, lions2007correlations}, clearly~$S$ is a non-negative local martingale, hence a supermartingale.
With $\tau_n := \inf \{t>0, B^H_t = n\}$ for $n\in\NN$, then
$$
S_0 =\EE\left[S_{T\land \tau_n}\right]
= \EE\left[S_{T}\mathbbm{\ind}_{\{T\leq \tau_n\}}\right]
+ \EE\left[S_{\tau_n}\mathbbm{\ind}_{\{\tau_n \leq T\}}\right],
$$
and therefore
$$
S_0 - \EE\left[S_{T}\right]
= \lim_{n\uparrow\infty} \EE\left[\EE\left[S_{\tau_n}\mathbbm{\ind}_{\{\tau_n \leq T\}}|\sigma(Y_{\beta})\right]\right].
$$
\blue{
Define on $\Ff_{t\wedge\tau_n}$ the measure $\widehat{\PP}_n$ (conditional on
$\sigma(Y_\beta)$) through the density
$$
\frac{\D\widehat{\PP}_n}{\D\PP}\bigg|_{\Ff_{t\wedge\tau_n}}
:= S_{t\wedge\tau_n}
= \exp\left\{\int_0^{t\wedge\tau_n}\sqrt{V_s}\big(\rho\D B_s + \overline{\rho}\D W_s\big)
   - \tfrac12\int_0^{t\wedge\tau_n} V_s\D s\right\},
$$
with $\overline{\rho}:=\sqrt{1-\rho^2}$
which is a $\PP$-martingale since $V$ is bounded on $[0,\tau_n]$ and $S_0=1$.
By Girsanov's theorem, the processes
$$
\widehat{B}^{(n)}_t := B_t - \int_0^{t\wedge\tau_n}\rho\sqrt{V_s}\D s
\qquad\text{and}\qquad
\widehat{W}^{(n)}_t := W_t - \int_0^{t\wedge\tau_n}\overline{\rho}\sqrt{V_s}\D s,
$$
are independent $\widehat{\PP}_n$-Brownian motions and hence
$$
\EE\Big[\EE\big[S_{\tau_n}\ind_{\{\tau_n\le T\}}\big|\sigma(Y_\beta)\big]\Big]
= S_0\EE\big[\widehat{\PP}_n(\tau_n\leq T)\big].
$$
The volatility driver $\Bh_t=\cf\int_0^t(t-s)^{\Hm}\D B_s$ reads, under $\widehat{\PP}_n$,
$$
\Bh_t = \widehat{B}^{H}_t
+ \rho\cf\int_0^{t\wedge\tau_n}(t-s)^{\Hm}\sqrt{V_s}\D s,
\qquad\text{with}\qquad
\widehat{B}^{H}_t := \cf\int_0^t (t-s)^{\Hm}\D\widehat{B}^{(n)}_s.
$$
Since $(t-s)^{\Hm}\sqrt{V_s}\ge 0$, the drift is non-positive precisely when
$\rho\le 0$ (and vanishes when $\rho=0$, where $\Bh\equiv\widehat{B}^H$); hence
$\Bh_t\le\widehat{B}^H_t$ for all $t\le\tau_n$ a.s.,} so $\tau_n\ge\tau_n^0:=\inf\{t>0:\widehat{B}^H_t=n\}$. By Dominated Convergence and since $\widehat{B}^{(n)}$ is a $\widehat{\PP}_n$-Brownian motion, then
$$
\lim_{n\uparrow\infty} \EE\left[\widehat{\PP}_n \left(\tau_n^0 \leq T \right)\right] = \EE\left[\lim_{n\uparrow\infty} \widehat{\PP}_n \left(\tau_n^0 \leq T \right)\right] = \EE\left[\lim_{n\uparrow\infty}\PP \left(\sup_{t\in[0,T]} B^{H}_t \geq n \right)\right] = 0.
$$
Therefore
$S_0 - \EE[S_T]= 0$,
which concludes the proof.
\end{proof}

In the context of option pricing, in particular for American options, 
the following result, adapted from~\cite{gerhold2026integrability}, 
guarantees that the model can be used:
\begin{lemma}
For any $t\geq 0$,
$\EE\left[\sup_{u \in [0,t]}S_u\right] < \infty$.
\end{lemma}
%}

\begin{proof}
In light of Theorem~\ref{thm:bbBmartingale}, 
we first assume that $\rho \leq 0$. 
Using the Riemann-Liouville fractional Brownian motion representation and Lemma~\ref{lem:gBm_decomposition},
the variance under~\eqref{eq:gBergomi_dynamics} reads
$$
V_t = \frac{\xi_0(t)}{\Ee_{\beta}(\cff t^{2H})}\exp\left\{\eta \cf\sqrt{Y_{\beta}} \int_0^t (t-s)^{\Hm} \D B_s\right\} =: f(t,Z_t),
$$
with 
$f(t,z) := \frac{\xi_0(t)\E^{z}}{\Ee_{\beta}(\cff t^{2H})}$,
$Z_t := \int_0^t K_{H, \eta}(t,s) \D B_s$ 
and $K_{H, \eta}(t,s):= \eta \cf\sqrt{Y_{\beta}} (t-s)^{\Hm}$. \blue{From~\cite[Theorem~A.4]{gerhold2026integrability} (strong existence and uniqueness for stochastic Volterra equations) applied as in the rough Bergomi example of~\cite[Section~3.1]{gerhold2026integrability}, 
we conclude that}
$$ 
\widetilde{Z}_t = 
Z_t + \int_0^t K_{H, \eta}(t,s) \rho \sqrt{\frac{\xi_0(s)}{\Ee_{\beta}(\cff s^{2H})}} \exp\left\{ \frac{\widetilde{Z}_s}{2}\right\} \D s
$$
admits a unique strong solution. 
Furthermore, since $f(t,\cdot)$ is non-decreasing for each~$t$,
$$
0 \leq f(t, \widetilde{Z}_t) \leq \frac{\xi_0(t) \exp\{Z_t\}}{\Ee_{\beta}(\cff t^{2H})}.
$$
The existence of a strong and unique solution to $\widetilde{Z}_t$ allows to use~\cite[Theorem~2.5]{gerhold2026integrability};
when coupled with the inequality above, we obtain for some $\kappa,\gamma>0$,
\begin{align}    \EE\left[\sup_{u\in[0,t]} S_u \Big\vert \sigma(Y_{\beta})\right] &\leq \kappa + \gamma \int_0^t \EE\left[f(s, \widetilde{Z}_s) \vert \sigma(Y_{\beta})\right] \D s\\
&\leq
\kappa + \gamma \int_0^t \frac{\xi_0(s)}{\Ee_{\beta}\left(\cff s^{2H}\right)}\exp\left\{\cff Y_{\beta}  s^{2H}\right\} \D s < \infty.
    \end{align}
    Taking expectations results in
    $$
    \EE\left[\sup_{u\in[0,t]} S_u\right] \leq \kappa + \gamma \int_0^t \EE\left[f(s, \widetilde{Z}_s)\right] \D s \leq \kappa + \gamma \int_0^t \xi_0(s)\D s < \infty.
    $$
\end{proof}

\blue{
\begin{remark}
The integrability bound above shows $\sup_{u\le t}S_u\in L^1$, so the non-negative local
martingale $S$ is therefore a true martingale. This gives an alternative proof of Theorem~\ref{thm:bbBmartingale} under the same hypothesis $\rho\le0$, independent
of the change-of-measure argument (the proof of the bound uses only the explicit form of $V$, not the martingale property).
\end{remark}
}

%%%%%%%%%%%%%%%%%%%%%%%%%%%%%%%%%%%%%%%%%%%%

%%%%%%%%%%%%%%%%%%%%%%%%%%%%%%%%%%%%%%%%%%%%
\subsection{VIX}\label{sec:gBergomiVIX}
Now that the basic properties for the stock price have been set, we move on to studying the VIX under the grey Bergomi model.

\subsubsection{VIX Dynamics}
The continuously monitored version of the VIX is defined as
$$
\VIX^2_{T} := \EE\left[\frac{1}{\Delta}\int_{T}^{T+\Delta}V_{s}\D s\middle| \FfB_T \right],
\qquad\text{for any }T\geq 0,
$$
with~$\Delta$ corresponding to one month.
To streamline the results, introduce the quantities
\begin{equation}\label{eq:Vprocesses}
\Vv_{s}^T := \displaystyle \int_0^T (s-u)^{\Hm}\D B_u
\qquad\text{and}\qquad
\Vv_{s,T} := \displaystyle \int_{T}^{s} (s-u)^{\Hm}\D B_u,
\end{equation}
for any $0\leq T\leq s$.
The following proposition derives an expression for it:
%\adri{Change: Addition of missing $\frac{1}{\Delta}$ factor (Reviewer 1, Comment 2 and Reviewer 2, Comment 9)}

\begin{proposition}\blue{\label{prop:vixdynamics}
The VIX dynamics under~\eqref{eq:gBergomi_dynamics} are given, for any $T\geq 0$, by 
$$
\VIX_T^2
= \frac{1}{\Delta}\int_T^{T+\Delta}
  \frac{\xi_0(s)}{\Ee_\beta(\cff s^{2H})}
  \exp\left\{\eta\cf\sqrt{Y_\beta}\Vv_s^T + \cff Y_\beta(s-T)^{2H}\right\}\D s.
$$
}\end{proposition}

\begin{proof}
{\color{blue}%
Let $\Vv_t := \Vv_t^t$. Since $Y_\beta$ and $\Vv_s^T=\int_0^T(s-u)^{\Hm}\D B_u$ are both
$\FfB_T$-measurable, while $\Vv_{s,T}=\int_T^s(s-u)^{\Hm}\D B_u$ is independent of $\FfB_T$
(future $B$-increments, and $Y_\beta\perp B$) and centred Gaussian with variance
$\tfrac{(s-T)^{2H}}{2H}$, then
\begin{align*}
\VIX_T^2
&= \frac{1}{\Delta}\int_T^{T+\Delta}
   \frac{\xi_0(s)}{\Ee_\beta(\cff s^{2H})}
   \EE\left[\exp\{\eta\cf\sqrt{Y_\beta}\Vv_s\}\middle|\FfB_T\right]\D s\\
&= \frac{1}{\Delta}\int_T^{T+\Delta}
   \frac{\xi_0(s)}{\Ee_\beta(\cff s^{2H})}
   \EE\left[\exp\{\eta\cf\sqrt{Y_\beta}(\Vv_s^T+\Vv_{s,T})\}\middle|\FfB_T\right]\D s.
\end{align*}
Taking the Gaussian expectation over $\Vv_{s,T}$ alone,
\begin{align*}
\EE\left[\exp\{\eta\cf\sqrt{Y_\beta}(\Vv_s^T+\Vv_{s,T})\}\middle|\FfB_T\right]
&= \exp\{\eta\cf\sqrt{Y_\beta}\Vv_s^T\}
   \EE\left[\exp\{\eta\cf\sqrt{Y_\beta}\Vv_{s,T}\}\middle|\FfB_T\right]\\
&= \exp\Big\{\eta\cf\sqrt{Y_\beta}\Vv_s^T
   + \tfrac12\eta^2\cf^2 Y_\beta\tfrac{(s-T)^{2H}}{2H}\Big\}\\
&= \exp\big\{\eta\cf\sqrt{Y_\beta}\Vv_s^T + \cff Y_\beta(s-T)^{2H}\big\}.
\end{align*}
Substituting back, the proof follows from
$$
\VIX_T^2
= \frac{1}{\Delta}\int_T^{T+\Delta}
  \frac{\xi_0(s)}{\Ee_\beta(\cff s^{2H})}
  \exp\big\{\eta\cf\sqrt{Y_\beta}\Vv_s^T + \cff Y_\beta(s-T)^{2H}\big\}\D s.
$$
}
\end{proof}
%%%%%%%%%%%%%%%%%%%%%%%%%%%%%
%%%%%%%%%%%%%%%%%%%%%%%%%%%%%
\subsubsection{VIX Futures}
Following~\cite{jacquier2018vix}, the VIX Future~$\Ffr_T$ with maturity~$T$ is given by
%\adri{Change: Change $X_s$ to $S_s$ (Reviewer 2, Comment 10)}
\blue{
\begin{equation}\label{eqn:vixfuture}
\Ffr_T := \EE\left[\VIX_T\right]
= \EE\left[\sqrt{\frac{1}{\Delta}\int_T^{T+\Delta} \EE\left[\D\langle \log S\rangle\middle\vert\FfB_T\right]}\right]
= \EE\left[\sqrt{\frac{1}{\Delta}\int_T^{T+\Delta} \xi_T(s) \D s }\right].
\end{equation}
}
Since $\xi_T(t) = \EE[V_t | \FfB_T]$ for $t\geq T$, the following is immediate from Proposition~\ref{prop:vixdynamics}:

\begin{proposition}\label{prop:forwardvarcurve}
Under~\eqref{eq:gBergomi_dynamics} the forward variance curve admits the representation
\blue{
$$
\xi_T(t) = \frac{\xi_0(t)}{\Ee_\beta(\cff t^{2H})}
\exp\left\{\eta\cf\sqrt{Y_\beta}\Vv_t^T + \cff Y_\beta(t-T)^{2H}\right\},
\qquad\text{for any } t \geq T.
$$
}
\end{proposition}

\subsubsection{Upper and lower bounds for VIX Futures}
Similarly to~\cite[Theorem~3.2]{jacquier2018vix}, one can construct bounds for VIX Futures.

\begin{proposition}\label{prop:vixbounds}
    The following bounds hold for VIX Futures:
    $$
    \frac{1}{\Delta}\int_T^{T+\Delta}
   \sqrt{\frac{\xi_0(s)}{\Ee_\beta(\cff s^{2H})}}\;
   \Ee_\beta\left(\frac{\cff}{4}\big(s^{2H}+(s- T)^{2H}\big)\right)\D s
    \leq 
    \Ffr_T 
    \leq 
    \sqrt{
    \frac{1}{\Delta}
    \int_T^{T+\Delta} \xi_0(s) \D s 
    }.
    $$
\end{proposition}

\begin{proof}
The conditional Jensen's inequality and Fubini's theorem give 
$$
\Ffr_T = \EE\left[\VIX_T\right] = \EE\left[\sqrt{\frac{1}{\Delta}\int_T^{T+\Delta} \xi_T(s) \D s }\right] \leq \sqrt{\frac{1}{\Delta}\int_T^{T+\Delta} \EE[\xi_T(s)] \D s }.
$$
For any $s>0$, the martingale property of~$(\xi_{t}(s))_{t\leq s}$ implies that
$\Ffr_T \leq \sqrt{\frac{1}{\Delta}\int_T^{T+\Delta} \xi_0(s) \D s }$.
To obtain the lower bound we use Proposition~\ref{prop:forwardvarcurve}, %\adri{Change: Use jensen instead of cauchy-schwartz (Reviewer 2, Comemnt 14)}\blue{Jensen's inequality} and Fubini's theorem to deduce
\blue{
\begin{align}
\Ffr_T := \EE\left[\VIX_T\right] 
& = \EE\left[
\sqrt{
\frac{1}{\Delta}
\int_T^{T+\Delta} 
\frac{\xi_0(s)}{\Ee_\beta(\cff s^{2H})}
\exp\left\{\eta\cf\sqrt{Y_\beta}\Vv_s^T + \cff Y_\beta(s-T)^{2H}\right\}
\D s \;}\right]\\
&\geq \EE\left[\frac{1}{\Delta}\int_T^{T+\Delta} \sqrt{
\frac{\xi_0(s)}{\Ee_\beta(\cff s^{2H})}
\exp\left\{\eta\cf\sqrt{Y_\beta}\Vv_s^T + \cff Y_\beta(s-T)^{2H}\right\}
}
\D s \right]\\
&= \frac{1}{\Delta}
\int_T^{T+\Delta} 
\sqrt{
\frac{\xi_0(s)}
{ \Ee_{\beta}\left(\cff s^{2H}\right)}}
\EE\left[
\sqrt{
\exp\left\{\eta\cf\sqrt{Y_\beta}\Vv_s^T + \cff Y_\beta(s-T)^{2H}\right\}
}\right] \D s\\
&= \frac{1}{\Delta}\int_T^{T+\Delta}
   \sqrt{\frac{\xi_0(s)}{\Ee_\beta(\cff s^{2H})}}\;
   \Ee_\beta\left(\frac{\cff}{4}\left(s^{2H}+(s-T)^{2H}\right)\right)\D s
\end{align}
where the last step is the result of conditioning on $Y_{\beta}$ first and then computing the expectation using~\eqref{eq:MWrightMom}.
}

\end{proof}

Figure~\ref{fig:boundscenarios} illustrates these upper and lower bounds 
when
$(H, \beta, \eta) = (0.065,0.037,0.65)$ (consistent with our calibration results) in the following three scenarios for the initial forward variance curve (same as in~\cite{jacquier2018vix}):
\begin{equation}\label{eq:Scenarios}
\begin{array}{rl}
\text{Scenario 1: } & \xi_0(t) = 0.235^2, \\
\text{Scenario 2: } & \xi_0(t) = 0.235^2 (1+t)^2, \\ 
\text{Scenario 3: } & \xi_0(t) = 0.235^2 \sqrt{1+t}.
\end{array}
\end{equation}

\begin{figure}[ht]
\centering
    \subfloat{\includegraphics[width=0.5\linewidth]{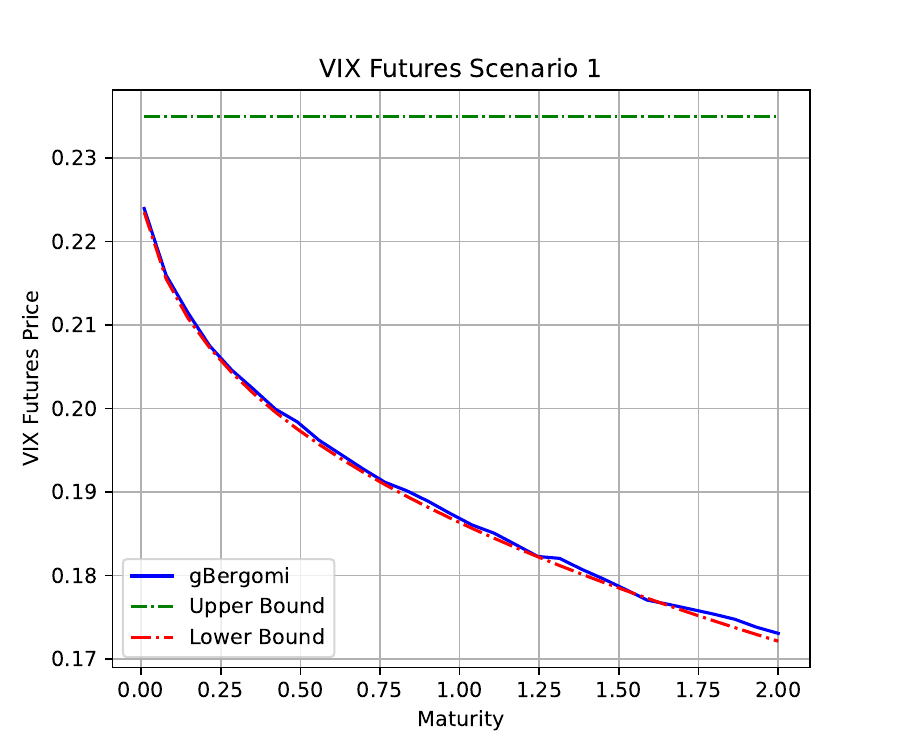}}
    \subfloat{\includegraphics[width=0.5\linewidth]{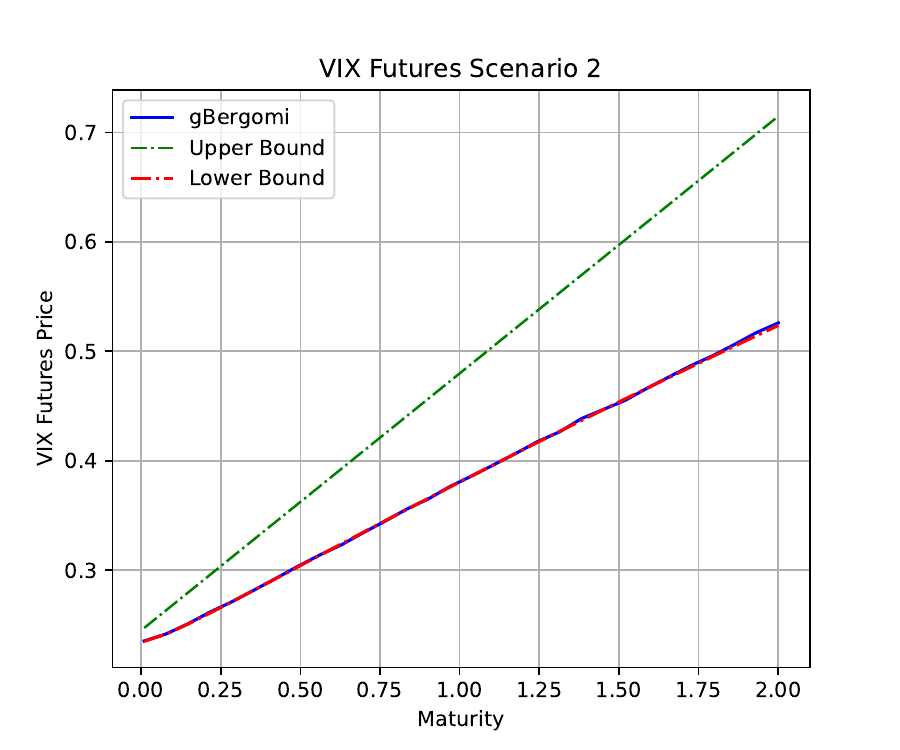}}
    \newline 
    \includegraphics[width=0.5\linewidth]{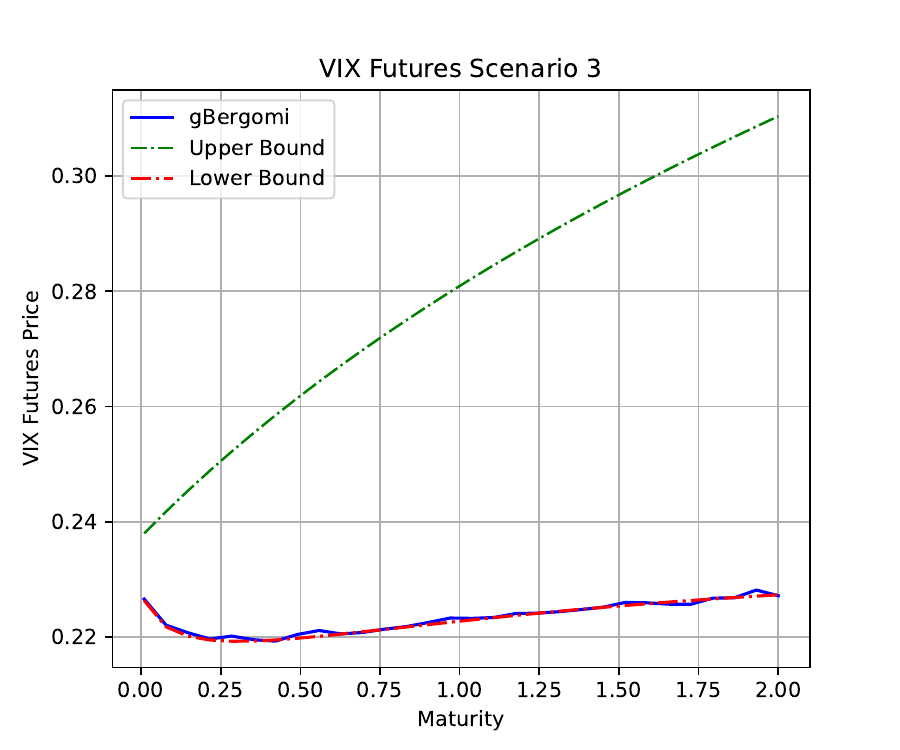}
    \caption{Upper and lower bounds in all three scenarios.}
\label{fig:boundscenarios}
\end{figure}

\blue{Our lower bound here is as tight as the one in~\cite[Theorem~3.2]{jacquier2018vix} for the same reason therein. 
The exact same argument in~\cite{jacquier2018vix} still holds: consider a simplified and deterministic version of the
VIX futures price in Proposition~\ref{prop:vixbounds}, denoted by
$$
\phi(T) \coloneqq \sqrt{\frac1\Delta \int_T^{T+\Delta}f(t)\D t} = \sqrt{f(T) + \frac\Delta2 f'(T) + \frac{\Delta^2}{6} f''(T) + \Oo(\Delta^3)}
$$
for some strictly positive (deterministic) function $f \in C^2(\RR)$. Using the Jensen induced lower bound we obtain
$$
\psi(T)\coloneqq \frac1\Delta \int_T^{T+\Delta} \sqrt{f(t)}\D t = \sqrt{f(T)} + \frac \Delta4 \frac{f'(T)}{\sqrt{f(T)}} + \Oo(\Delta^2).
$$
Taking the difference of the squares yields
$$
\phi(T)^2 - \psi(T)^2 = \Delta^2 \left(\frac{f''(T)}{6} - \frac{f'(T)}{16 f(T)}\right) + \Oo(\Delta^3).
$$
Hence, for small $\Delta$ 
(such as in VIX futures) the lower bound is close to the
original value which explains 
(at least for the deterministic case) the behaviour observed in Figure~\ref{fig:boundscenarios}.
}

\subsection{Skew-Stickiness Ratio (SSR)}
\blue{While joint SPX/VIX calibration is the primary objective, a model's term-structure dynamics should also be assessed. A standard diagnostic is the Skew-Stickiness Ratio (SSR)
of Bergomi~\cite{BergomiIV} (see also~\cite{florian2024smile}), which measures how the
at-the-money implied volatility skew responds to movements in the underlying. 
Let $[X,Y]_{[t,s]}$ denote the quadratic covariation of two It\^o processes $X,Y$ over $[t,s]$,
and let $\sigma_t(\tau)$ and $\Ss_t(\tau)$ denote, respectively, the at-the-money-forward
implied volatility and its skew at time~$t$ for maturity~$\tau$. The SSR at time~$t$ and
maturity~$\tau$ is defined as
$$
\mathcal{R}_t(\tau) := \frac{1}{\Ss_t(\tau)}
\frac{\dv{s}[\log S,\sigma_{\cdot}(\tau)]_{[t,s]}\big|_{s=t}}
     {\dv{s}[\log S]_{[t,s]}\big|_{s=t}},
$$
that is, the instantaneous covariation between the log-price and the ATM volatility,
divided by the instantaneous variance of the log-price and normalised by the skew.
Replacing $\sigma$ with the variance-swap implied volatility $\sigma_{\mathrm{VS}}$ gives the
variance-swap SSR~\cite{florian2024smile}. 
For the gBergomi model both versions share the same short-maturity limit.}
\begin{proposition}
For the gBergomi model~\eqref{eq:gBergomi_dynamics},
$\displaystyle \lim_{\tau\downarrow0}\mathcal{R}_t(\tau) = H + \tfrac32.$
\end{proposition}
The variance-swap case follows exactly as in~\cite[Proposition~20]{florian2024smile}, and the
standard case from~\cite[Corollary~5.6]{friz2024computingssr}.

%%%%%%%%%%%%%%%%%%%%%%%%%%
%%%%%%%%%%%%%%%%%%%%%%%%%%
\section{Numerical algorithms}\label{sec:algos}
\subsection{Numerical implementation of VIX process}\label{sec:AlgoVIX}
%\adri{Change: added more details on the simulation scheme (Reviewer 1, Comment 8)}
\blue{Simulating the VIX under~\eqref{eq:gBergomi_dynamics} requires generating the Volterra process~$\Vv^T$ in~\eqref{eq:Vprocesses} on a discrete grid~$\mathfrak{T} = \{\tau_j\}_{j=0,\ldots,N}$ over the interval~$[T, T+\Delta]$.
    Since $\Vv^T$ is a centred Gaussian process, its law on the grid is fully determined by the covariance matrix~$C_{jk} := \EE[\Vv^T_{\tau_j}\Vv^T_{\tau_k}]$, which can be computed explicitly via~\eqref{eqn:covstruct}. 
    The na\"ive approach -- computing a full Cholesky factorisation of the $N\times N$ matrix~$C$ -- has $\Oo(N^3)$ complexity and, more critically, becomes numerically unstable for small values of~$H$ due to the near-singularity of~$C$.
    To circumvent this, we adopt the \emph{truncated Cholesky} scheme of~\cite{jacquier2018vix}, which proceeds in two stages.
    In the first stage, the process is simulated exactly at the 
    first~$l$ grid points $\{\tau_1,\ldots,\tau_l\}$ 
    via a standard Cholesky decomposition of the $l\times l$ 
    leading sub-matrix of~$C$. This small exact block absorbs the ill-conditioning that arises from grid points close to~$T$, 
    where the kernel $(t-s)^{\Hm}$ is singular. In the second stage, the remaining points $\{\tau_{l+1},\ldots,\tau_N\}$ are generated recursively using a one-step Markovian approximation: 
    each $\Vv^T_{\tau_j}$ is drawn conditionally on $\Vv^T_{\tau_{j-1}}$ alone, exploiting the fact that, sufficiently far from~$T$, the correlation between 
    consecutive points is close to one and the conditional 
    distribution is well approximated by a bivariate Gaussian 
    projection. 
    Concretely, writing~$\rho_{j-1,j} := \mathrm{Corr}(\Vv^T_{\tau_{j-1}}, \Vv^T_{\tau_j})$
    for $j = l+1,\ldots,N$, the recursion reads
    $$
    \Vv^T_{\tau_j} = \sqrt{\VV[\Vv^T_{\tau_j}]}
    \left(
      \frac{\rho_{j-1,j}\Vv^T_{\tau_{j-1}}}{\sqrt{\VV[\Vv^T_{\tau_{j-1}}]}}
      + \sqrt{1 - \rho_{j-1,j}^2}\; Z_j
    \right),
    \qquad Z_j \overset{\mathrm{iid}}{\sim} \Nn(0,1).
    $$
    This reduces the per-path cost of the second stage 
    to~$\Oo(N)$ while retaining the exact covariance 
    structure at each consecutive pair.
}
\begin{remark}
We set $l=8$ as in~\cite{jacquier2018vix} to avoid numerical issues with a small determinant.
\end{remark}
% Our numerical implementation here closely follows the methodology developed in~\cite{jacquier2018vix},
% and we thus refer the reader to the latter for full details.
Recall that the covariance structure
of~$\Vv^T$ in~\eqref{eq:Vprocesses} reads
\begin{align}\label{eqn:covstruct}
\EE\left[\Vv_t^T \Vv_s^T\right]
&= \int_0^T [(t-u)(s-u)]^{\Hm} \D u\\
&= \frac{(s-t)^{\Hm}}{\Hp}\left\{t^{\Hp}\Fm\left(\frac{-t}{s-t}\right) - (t-T)^{\Hp}\Fm\left(\frac{T-t}{s-t}\right)\right\},
\end{align}
for any $t<s$, where~${}_2F_1$ is the Hypergeometric function~\cite[Chapter 15]{handbookoffunctions} and
$$
\Fm(u):= {}_2F_1\left(-\Hp, \Hp, 1+\Hp,u\right).
$$
\blue{The full simulation procedure is summarised below.}
\algo[VIX simulation]{
Fix a grid $\mathfrak{T} = \{\tau_j\}_{j = 0,\ldots, N}$ on $[T,T+\Delta]$, and $l\in\mathbb{N}$.
%\adri{Change: updated the index in the trapeziodal rule, correctly defined the quantity $Q_{T,t}$, and changed step (iii) to use a set of i.i.d. random variables}
\blue{
\begin{enumerate}[(i)]
    \item Compute the covariance matrix of 
$(\Vv_{\tau_j}^T)_{j=1,\ldots,l}$ using~\eqref{eqn:covstruct};
\item compute $\rho_{j-1,j} := \mathrm{Corr}(\Vv_{\tau_{j-1}}^T,\Vv_{\tau_{j}}^T)$ by Cholesky decomposition for $j = l+1, \ldots, N$;
    \item generate $\{\Vv_{\tau_j}^T\}_{j=l+1, \ldots, N}$ via
    $$
    \Vv_{\tau_j}^T = \sqrt{\VV[\Vv_{\tau_j}^T]} \left(\frac{\rho_{j-1,j}\Vv_{\tau_{j-1}}^T}{\sqrt{\VV[\Vv_{\tau_{j-1}}^T]} } + \sqrt{1 - \rho_{j-1,j}^2} \ Z_j\right), \textup{ for } j = l+1, \ldots, N,
    $$
    where $\{Z_j\}_{j=l+1,\ldots,N}$ are centred i.i.d. Gaussian with unit variance;
    \item compute $\VIX_T$ by numerical integration, for example with a trapezoidal rule:
    $$
    \text{VIX}_T \approx \sqrt{\frac{1}{\Delta} \sum_{j=0}^{N-1} \frac{Q^2_{T,\tau_j}+Q^2_{T,\tau_{j+1}}}{2} (\tau_{j+1} - \tau_{j})},
    $$
    where $\displaystyle Q^2_{T,t} \coloneqq \frac{\xi_0(t)}{\Ee_\beta(\cff t^{2H})}
    \exp\left\{\eta\cf\sqrt{Y_\beta}\Vv_t^T + \cff Y_\beta(t-T)^{2H}\right\}.$
\end{enumerate}
}
}

Figure~\ref{fig:mccholesky} shows the results for $10^5$ Monte Carlo simulations in Scenario~1 in~\eqref{eq:Scenarios} for the Truncated Cholesky scheme. As a comparison, VIX Futures prices for the rough Bergomi (where $\beta = 1$) and the Monte Carlo standard deviations are given. 
Compared to rough Bergomi~\cite{jacquier2018vix}, its grey counterpart (with $\beta \in \{0.1, 0.5\}$) yields lower prices, \blue{consistent with a Jensen-type effect arising from the concavity of the square-root transformation in the VIX definition.}

\begin{figure}[ht]
\centering
    \subfloat{\includegraphics[width=0.5\linewidth]{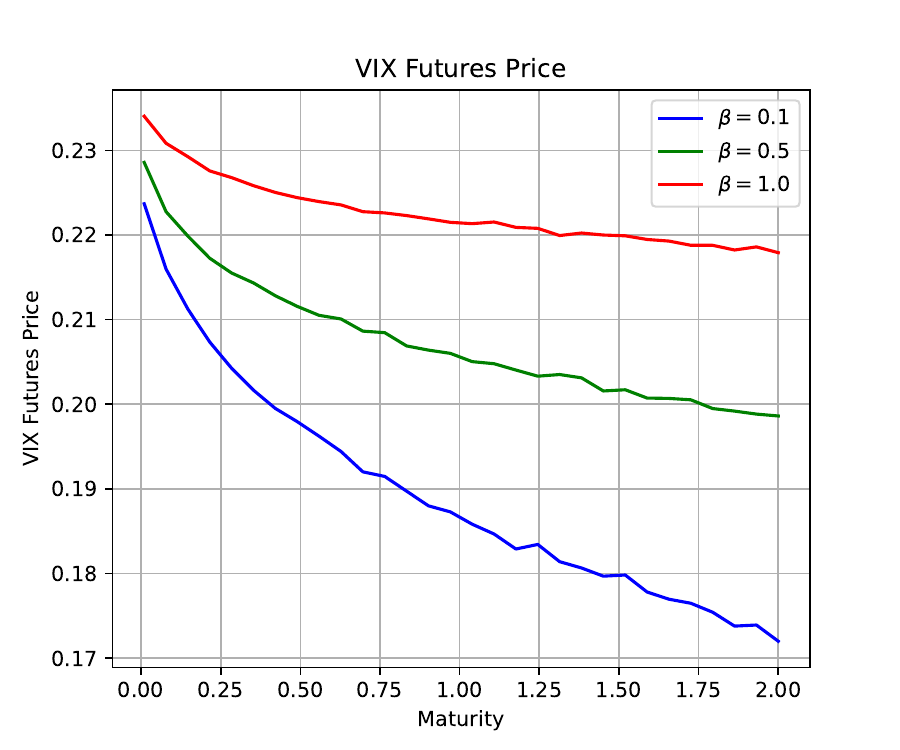}}
    \subfloat{\includegraphics[width=0.5\linewidth]{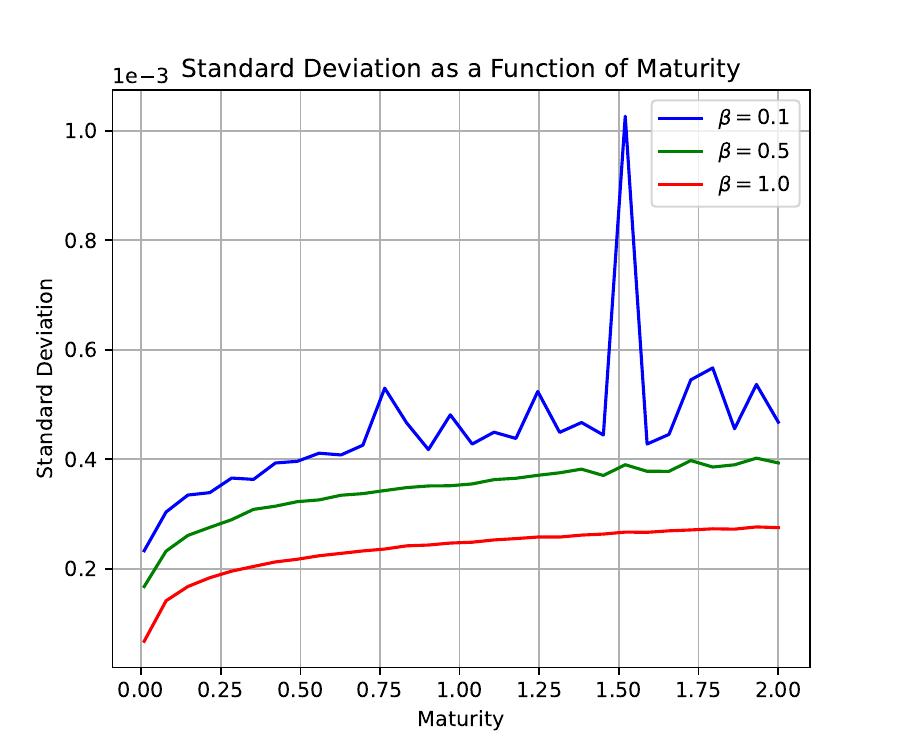}}
    %\newline 
    %\subfloat{\includegraphics[width=0.5\linewidth]{mc_std_sim.pdf}}
    %\subfloat{\includegraphics[width=0.5\linewidth]{mc_time.pdf}}
    \caption{Truncated Cholesky Monte Carlo prices and Monte Carlo standard deviations for $\beta \in \{0.1, 0.5, 1.0\}$.}
    \label{fig:mccholesky}
\end{figure}

%%%%%%%%%%%%%%%%%%%%%%%%%%%%%%%%%%%%%%%%%%%%
%%%%%%%%%%%%%%%%%%%%%%%%%%%%%%%%%%%%%%%%%%%%

\subsection{Algorithm for SPX options}\label{sec:vixtospx}
Based on the above algorithm for the VIX, 
we now develop a numerical scheme for option prices under~\eqref{eq:gBergomi_dynamics}.
From the definition of~$\Vv_t$ in Proposition~\ref{prop:vixdynamics} (and the first line of its proof), then
$\VV[\Vv_t] = \frac{1}{2H}t^{2H}$ and
\begin{align*}
\EE[\Vv_t\Vv_s]
    = \frac{t^{\Hp} s^{\Hm}}{\Hp}{}_2F_1\left(-\Hm, 1, 1+\Hp, \frac{t}{s}\right),
    \qquad\text{for }t < s.
\end{align*}

\algo[Spot process simulation]{Fix $\kappa \geq 1$ and the grid $\Tt := \{t_i\}_{i=0,\ldots,n_T}$.
\begin{enumerate}[(i)]
    \item Simulate the Volterra process~$\Vv$ on~$\Tt$ using~\eqref{eqn:covstruct};
    \item \blue{sample $Y_{\beta}$ (see Appendix~\ref{apx:mwrightgen} for one-sided M-Wright variate generation) and construct the underlying grey driver~$\sqrt{Y_{\beta}}\Vv$ on~$\Tt$};%\adri{Change: add a variate generation step (Reviewer 1, Comment 9)}
    \item compute the variance~$V$ from~\eqref{eq:gBergomi_dynamics} on~$\Tt$; 
    \item back out the Brownian path from~$\Vv$ to obtain $\{B_{t_i}\}_{i=0}^{n_T-1}$;
    %\adri{Change: renamed $B^\perp$ to $W$ to match model notation (line~\eqref{eq:gBergomi_dynamics}), and inlined the correlation directly into the Euler scheme to avoid introducing a composite letter (Reviewer 2, Comment 19)}
    \item \blue{generate $\{W_{t_i}\}_{i=0}^{n_T-1}$ i.i.d.\ $\sim\mathcal{N}(0,t_{i+1}-t_i)$, independent of~$B$;}
    \item \blue{simulate $X := \log(S)$ using a forward Euler scheme
    $$X_{t_{i+1}} = X_{t_{i}} - \half V_{t_i}(t_{i+1} - t_i) + \sqrt{V_{t_i}} \bigl(\rho\Delta B_{t_i} + \sqrt{1-\rho^2}\Delta W_{t_i}\bigr), \quad\text{for } i=0,\ldots, n_T -1;$$}
    \item compute the expectation by averaging the payoff of all paths.
\end{enumerate}
}\label{algo:gBergomisim}

%\adri{Change: added remark explaining steps (iii) and (iv) of the algorithm. (Reviewer 2, Comment 19)}
\blue{
\begin{remark}
    In step~(iii), the variance $V_{t_i}$ is computed directly from the Volterra variates $\{\mathcal{V}_{t_i}\}$ and $Y_\beta$ already generated in steps~(i)--(ii) via
    $$V_{t_i} = \frac{\xi_0(t_i)}{\Ef_\beta(\cff t_i^{2H})}\exp\left\{\eta\sqrt{Y_\beta}\mathcal{V}_{t_i}\right\}.$$
    The Brownian increments $\{\Delta B_{t_i}\}$ in step~(iv) are the standard Gaussian variates produced during the Cholesky simulation of $\mathcal{V}$ in step~(i).
\end{remark}
}

\begin{remark}
    This is not the most  effective way to price since Cholesky is notoriously slow.
    One may instead consider a hybrid scheme approach (as done in~\cite{jacquier2018vix} using~\cite{Bennedsen_2017}) or a Markovian approximation, as explained in Appendix~\ref{apx:markovian_approximation}.
\end{remark}

To gain intuition about~$\beta$, consider Scenario~1 in~\eqref{eq:Scenarios} with $(H, \eta) = (0.065, 0.65)$,
$10^5$ Monte Carlo simulations, $\beta \in \{0.1, 0.5, 1\}$ and $\rho \in \{-0.9, -0.3, 0.0\}$. 
The prices of Call options and implied volatilities on~$S$ with maturity $T = 1$ can be observed in Figure~\ref{fig:ivgBergomi}. 

\begin{figure}[ht]
\centering
    \subfloat{\includegraphics[width=1\linewidth]{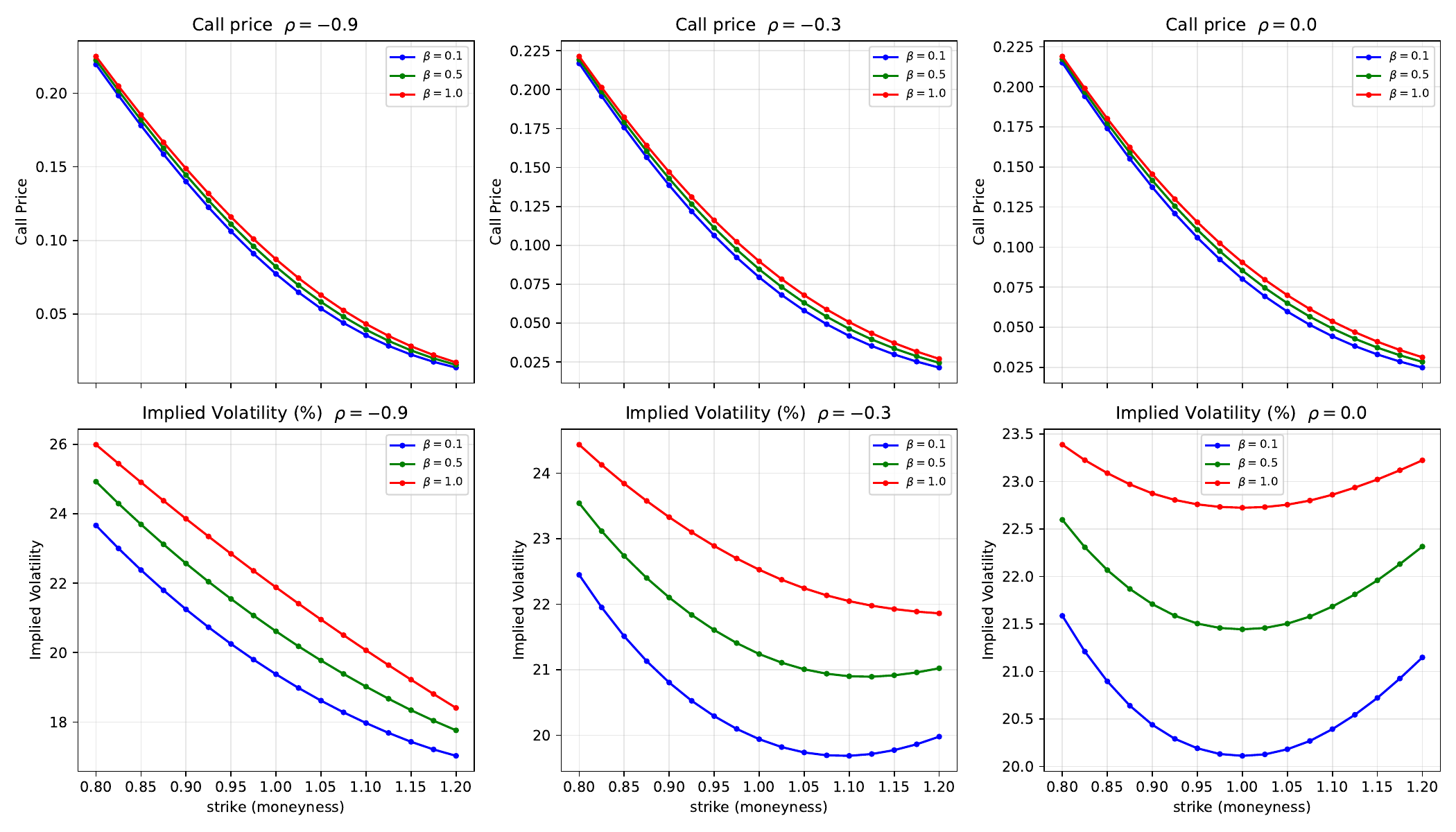}}
    \caption{SPX Call option prices and implied volatilities.}
\label{fig:ivgBergomi}
\end{figure}

\blue{
\begin{remark}
Figure~\ref{fig:ivgBergomi} shows the SPX smile for
$\rho\in\{-0.9,-0.3,0.0\}$ and $\beta\in\{0.8,0.9,1.0\}$. The two parameters play distinct roles. The correlation $\rho$ controls the skew: strongly negative correlation produces a steep downward-sloping smile, while $\rho=0$ gives an
approximately symmetric smile. By contrast, $\beta$ mainly affects the level and convexity of the implied volatility curve. 
Increasing~$\beta$ shifts the smile upward in a nearly parallel way, with little effect on the skew.
\end{remark}}

%%%%%%%%%%%%%%%%%%%%%%%%%%%%%%%%%%%%%%%%%%%%
%%%%%%%%%%%%%%%%%%%%%%%%%%%%%%%%%%%%%%%%%%%%

\section{Asymptotics of the SPX \& VIX smiles}\label{sec:asymptotics}
We now develop asymptotic closed-form expressions for the implied volatility smile of the SPX and the VIX,
in particular for the ATM short-time level, skew, and curvature, 
following the approach in~\cite[Chapter~6-8]{Alos2021MalliavinFinance}.
We proceed as in~\cite{Jacquier2021RoughOptions} and consider a square-integrable strictly positive process $\left\{A_t\right\}_{t \in [0,T]}$, adapted to the filtration~$\FF$ introduced in Section~\ref{sec:gBergomi}. 
We further introduce the $\FF$-martingale conditional expectation process
$$
\Ffr_{t,T}^A\coloneqq\EE\left[A_T|\Ff_t\right], \quad \text { for all } t \in [0,T],
$$
which is nothing else than a time $t$-price of a Future contract on~$A_T$. 
%\adri{Change: replaced $W^i$ with $Z^i$ throughout to avoid overloading the model notation $W$ (stock BM). (Reviewer 2, Comment 21)}
\blue{We use~$\DD$ to denote the domain of the Malliavin operators~$\Dr^i$ for $i\in\{1,2\}$ with respect to the Brownian motion~$Z^i$
and write $\LL^2 \coloneqq L^2([0, T]; \DD)$
(and refer the interested reader to~\cite{DavidNualart2006TheTopics} for an in-depth introduction to Malliavin calculus).} Assuming $A_T \in \DD$, the Clark-Ocone formula~\cite[Theorem~1.3.14]{DavidNualart2006TheTopics} reads, 
for each $t \in [0, T]$,
$$
\Ffr_{t,T}^A=\EE\left[\Ffr_{t,T}^A\right] + \sum_{i=1}^2\int_0^t \ff_{s}^{i}(A_T) \D W_s^i,
$$
where $Z^1\coloneqq W$ and $Z^2\coloneqq B$ and $\ff_{s}^{i}(A_T)\coloneqq \EE[\Dr^i_s A_T | \Ff_s]$.
%\adri{Change: replaced $N$ with $2$ in the martingale representation and in $\boldsymbol{\phi}$, consistent with the two Brownian motions $Z^1 = W$ and $Z^2 = B$ defined above. (Reviewer 2, Comment 22; Reviewer 1, Comment 19)}
Now, since $\Ffr_{t,T}^A$ is an $\FF$-martingale, this can be further written as
$$
\blue{\Ffr_{t,T}^A = \Ffr_{0,T}^A + \sum_{i=1}^2\int_0^t \Ffr_{s,T}^A \phi_s^i \D Z_s^i, \quad \text{ with } \quad \phi_s^i \coloneqq \frac{\ff_{s}^{i}(A_T)}{\Ffr_{s,T}^A},}
$$
which is well defined since~$A$ is strictly positive,
and hence so is ${\Ffr_{s,T}^A}$. Finally, for \blue{$\boldsymbol{\phi}\coloneqq (\phi_1, \phi_2) \in \LL^{2}$}, 
we define (this will be required in Proposition~\ref{prop:gBergomi_VIX_asym})
%\adri{Change: renamed $u_t$ to $h_t$ to avoid clash with $u$ used as a Malliavin differentiation index. (Reviewer 2, Comment 23)}
\begin{equation}
    \blue{h_t \coloneqq \sqrt{\frac{1}{T-t} \int_t^T \|\boldsymbol{\phi}_s\|^2 \D s},  \quad\text{ for } t\in[0,T).}
\end{equation}
Since $\Ffr^A_{\cdot, T}$ is a martingale, derivative contracts on this process do not exhibit arbitrage under the pricing measure, thus the fair price of a European Call with maturity~$T$ and log-strike $k\in \RR$ can be written as
$$
\Call_t(k) \coloneqq \EE\left[\left(\Ffr_{T,T}^A - \E^k\right)^+|\Ff_t\right] = \EE\left[\left(A_T - \E^k\right)^+|\Ff_t\right].
$$
Denote by $\BS(t, x, k, \sigma)$ the Black-Scholes price of a European Call option at time~${t\in[0, T]}$, with maturity $T$, log-price $x, \log$-strike~$k$ and volatility~$\sigma$, so that
$$
\BS(t, x, k, \sigma)= \begin{cases}\E^x \Nn\left(d_{+}(x, k, \sigma)\right)-\E^k \Nn\left(d_{-}(x, k, \sigma)\right), & \text {if } \sigma \sqrt{T-t}>0, \\ \left(\E^x-\E^k\right)^{+}, & \text {if } \sigma \sqrt{T-t}=0,\end{cases}
$$
with $d_{\pm}(x, k, \sigma):=\frac{x-k}{\sigma \sqrt{T-t}} \pm \frac{\sigma \sqrt{T-t}}{2}$ and~$\Nn$ the Gaussian cumulative distribution function.
\begin{definition}\
\begin{enumerate}[(i)]
    \item For $k \in \RR$, the implied volatility $\Ii_T(k)$ is the unique non-negative solution to $\Call_0(k)=$ $\BS(0, \log\Ffr_0^T, k, \Ii_T(k))$; we drop~$k$ in the at-the-money case $k=\log\Ffr_0^T$.
    \item The at-the-money implied skew $\Ss$ and curvature $\Cc$ at time zero are defined as
    $$
    \Ss_T\coloneqq\left.\partial_k \Ii_T(k)\right|_{k=\log\Ffr_0^T} \quad \text { and } \quad \Cc_T:=\left.\partial_k^2 \Ii_T(k)\right|_{k=\log\Ffr_0^T}
    $$
\end{enumerate}
\end{definition}
Using the decomposition property in Lemma~\ref{lem:gBm_decomposition} we can rewrite the variance process of the generalised grey Brownian motion in terms of the Riemann-Liouville fBm:
$$
    V_t = \frac{\xi_0(t)}{\Ef_{\beta}(\cff t^{2H})}\exp\left\{\eta \cf \sqrt{Y_{\beta}} \int_0^t (t-s)^{H_-}\D B_s\right\} 
$$
where $B$ is the standard Brownian motion related to~$B^{H}$. 
%\adri{Change: replaced "with respect to $B$" with "with respect to $Z^2$" for consistency with the notation $Z^2 \coloneqq B$ established in the Malliavin setup (Reviewer 2, Comment 24).}
\blue{Since the first integral term is $\FfB_t$-measurable, its Malliavin derivative with respect to~$Z^2$ is null for all $t\geq 0$.} 
We thus proceed as in~\cite[Section~5.6]{Alos2021MalliavinFinance} and  compute the following Malliavin derivatives with respect to~$B$, for $s\leq u\leq r\leq T\leq t$: 
\begin{equation}\label{eq:malliavin_gBergomi}
\begin{array}{r@{\;}l}
\Dr_r V_t & = \displaystyle \eta \cf \sqrt{Y_\beta} V_t (t-r)^{\Hm}, \\ 
\Dr_u \Dr_r V_t & = \displaystyle (\eta \cf)^2 Y_\beta V_t (t-r)^{\Hm}(t-u)^{\Hm}, \\
\Dr_s \Dr_u \Dr_r V_t & = \displaystyle (\eta \cf)^3 Y_\beta^{\frac{3}{2}} V_t (t-r)^{\Hm}(t-u)^{\Hm} (t-s)^{\Hm}.
\end{array}
\end{equation}
\subsection{Small-time VIX asymptotics}\label{sec:VIX_asymptotics}
%\adri{Change: moved remark here from before the Malliavin derivatives. Addresses (i) joint probability space, (ii) $Y_\beta$ as Malliavin parameter, (iii) $\VIX_T\in\DD$ chain including the missing conditional expectation step. (Reviewer 1, Comments 11 and 13)}
\blue{
Throughout this section, $W$, $B$ and $Y_\beta$ are mutually independent, defined on the product space $(\Omega,\Ff,\PP) = (\Omega_W\times\Omega_B\times\Omega_Y, \Ff_W\otimes\Ff_B\otimes\Ff_Y, \PP_W\otimes\PP_B\otimes\PP_Y)$;
$W$ and~$B$ are standard Brownian motions and~$Y_\beta$ is a one-sided $\M$-Wright random variable.
Malliavin differentiation is taken only in the directions $Z^1 = W$ and $Z^2 = B$; since $Y_\beta$ is independent of both Brownian filtrations, 
it is $\Ff_0$-measurable and satisfies $\Dr^i Y_\beta = 0$ for $i=1,2$.
}
\blue{
Regarding membership in $\DD^{1,2}$: 
\begin{enumerate}
    \item $Z_t := \int_0^t(t-s)^{\Hm}\D B_s$ belongs to $\DD^{1,2}$ as a square-integrable Wiener functional;
    \item $V_t \in \DD^{1,2}$ follows by the chain rule~\cite[Proposition~1.2.3]{DavidNualart2006TheTopics}; %applied to the exponential, with $Y_\beta$ treated as a parameter;
    \item $\xi_T(s) = \EE[V_s \vert \FfB_T] \in \DD^{1,2}$: since $Y_\beta$ is an independent, $\Ff_0$-measurable parameter with $\Dr_r Y_\beta = 0$, the commutativity of the Malliavin derivative with conditional expectation over the Brownian filtration~\cite[Proposition~1.2.6]{DavidNualart2006TheTopics} extends to the enlarged filtration $\FfB_T = \sigma(Y_\beta)\vee\Ff^B_T$, giving, for $r\leq T$,
    $$
    \Dr_r \EE[V_s \vert \FfB_T] = \EE[\Dr_r V_s \vert \FfB_T];
    $$
    \item $\VIX_T^2 = \frac{1}{\Delta}\int_T^{T+\Delta}\xi_T(s)\D s \in \DD^{1,2}$ by linearity;
    \item $\VIX_T\in\DD^{1,2}$ follows from~\cite[Proposition~1.2.3]{DavidNualart2006TheTopics} since~$\sqrt{\cdot}$
    is Lipschitz on $(0,\infty)$ and $\VIX_T > 0$ almost surely.
\end{enumerate}
}

\blue{Since $\VIX_T \in \DD^{1,2}$ by the preceding remark, the} Clark-Ocone formula~\cite[Theorem~1.3.14]{DavidNualart2006TheTopics} reads, for each $t \in [0, T]$,
$$
\Ffr_{t,T}^{\VIX}=\EE\left[\Ffr_{t,T}^{\VIX}\right] + \int_0^t \ff_{s}(\VIX_T) \D B_s,
$$
where $\ff_{s}(\VIX_T)\coloneqq \EE\left[\Dr_s\VIX_T \mid \FfB_s\right]$. 
Since $\Ffr_{s,T}^{\VIX} \neq 0$ almost surely, let \blue{$\phi^2_s\coloneqq\ff_{s}(\VIX_T)/\Ffr_{s,T}^{\VIX}$ and $\phi^1_s\coloneqq 0$}, so that \blue{$\boldsymbol{\phi} = (0, \phi^2)\in\LL^2$} and the representation takes the form of the general framework above with only the $Z^2 = B$ component active.

\begin{proposition}\label{prop:gBergomi_VIX_asym}
The following behaviours hold
with $J_1,J_2,J_3$ in~\eqref{eq:J1}-\eqref{eq:J2}-\eqref{eq:J3}:
\begin{equation}
\begin{aligned}
\lim _{T \downarrow 0} \Ii_T & = \frac{J_1}{2 \Delta}\frac{1}{\VIX_0^2}, & \text{ if } H \in\left(0, \half\right), \\
\lim _{T \downarrow 0} \Ss_T & =
\frac{J_2}{2J_1} - \frac{J_1}{2 \Delta}\frac{1}{\VIX_0^2}, & \text{ if } H \in\left(0, \half\right), \\
\lim _{T \downarrow 0} 
T^{\half - 3H} \Cc_T &= \frac{2 \Delta\VIX_0^2}{3J_1^2}
\lim _{T \downarrow 0}
\left\{T^{\half - 3H} J_3(T)\right\}, & \text{ if } H \in\left(0, \frac{1}{6}\right).
\end{aligned}
\end{equation}
\end{proposition}
\begin{proof}
From~\cite[Proposition~1]{Jacquier2021RoughOptions},
we know that the result holds with
$$
J_1\coloneqq\int_0^{\Delta} \EE\left[\Dr_0 V_r\right] \D r, \quad J_2\coloneqq\int_0^{\Delta} \EE\left[\Dr_0 \Dr_0 V_r\right] \D r \quad J_3(T)\coloneqq\int_T^{T+\Delta} \EE\left[\Dr_0 \Dr_0 \Dr_0 V_r\right] \D r.
$$
provided that the following assumptions, given in the above reference, hold:
for ${H\in(0, \half)}$, there exists $R \in L^p$ for all $p>1$ such that,
for all $t \leq s \leq u \leq T \leq r$,
\begin{enumerate}[(i)]
\item $(\Ffr^{\VIX}_{t,T})^{-1} \leq R$ almost surely;
\item almost surely,
\begin{enumerate}[a)]
\item $V_r \leq R$, 
\item $\Dr_u V_r \leq R (r-u)^{\Hm}$,
\item $\Dr_s \Dr_u V_r \leq R(r-s)^{\Hm}(r-u)^{\Hm}$,
\item $\Dr_t \Dr_s \Dr_u V_r \leq R(r-t)^{\Hm}(r-s)^{\Hm}(r-u)^{\Hm}$.
\end{enumerate}
\item \blue{$\EE\left[h_s^{-p}\right]$} is uniformly bounded in~$s$ and~$T$ for all $p>1$;
\item $u \mapsto \Dr_u V_r$, $s \mapsto \Dr_s \Dr_u V_r$, $t \mapsto \Dr_t \Dr_s \Dr_u V_r$ are almost surely continuous around zero.
\end{enumerate}
%\adri{Change: expanded R construction to address measurability, $L^p$ integrability, and continuity of Malliavin derivatives explicitly. Reviewer~1, Comment~14.}
In gBergomi,
(i) and~(iii) hold by Lemmas~\ref{lem:AssOneAsy} and~\ref{lem:AssThreeAsy}.
\blue{From~\eqref{eq:malliavin_gBergomi}, the Malliavin derivatives of~$V_r$ are of the form $(\eta\cf)^k Y_\beta^{k/2} V_r$ times a deterministic kernel, so the choice
$$
R := \esssup_{r\in[0,T]}\sum_{k=0}^3(\eta\cf)^k Y_\beta^{k/2} V_r
$$
satisfies the bounds in~(ii) almost surely.
Since each $V_r$ is $\FfB$-adapted and $Y_\beta$ is $\FfB_0$-measurable, $R$ is $\FfB_T$-measurable.
The integrability $R\in L^p$ for all $p\geq 1$ follows from the moment generating function~\eqref{eq:ggBmMGF} of the grey Brownian motion together with the M-Wright moment formula~\eqref{eq:MWrightMom}, which give finite exponential moments for $\sup_r V_r$ and for powers of $Y_\beta$.
Finally, the continuity statement in~(iv) holds since, by~\eqref{eq:malliavin_gBergomi}, $u\mapsto \Dr_u V_r = \eta\cf\sqrt{Y_\beta}V_r(r-u)^{\Hm}$ is continuous for $u<r$, and the same argument applies to the higher-order derivatives.}
We now derive explicit expressions for $J_1, J_2$ and $J_3(T)$:
\begin{align}\label{eq:J1}
    J_1 &=\int_0^{\Delta} \EE\left[\Dr_0 V_r\right] \D r = \int_0^{\Delta} \EE\left[\eta \cf \sqrt{Y_\beta} V_r r^{\Hm}\right] \D r \nonumber\\
    &= \xi_0 \eta \cf \int_0^{\Delta} \EE\left[ \sqrt{Y_\beta}\EE\left[ \Eeb\left(\eta \sqrt{Y_\beta} \Bh_r\right) \middle \vert  \sigma(Y_\beta)\right]\right] r^{\Hm} \D r \nonumber\\
    &= \xi_0 \eta \cf \int_0^{\Delta} \EE\left[ \sqrt{Y_\beta}\EE\left[ \exp\left\{\eta\sqrt{Y_\beta} \Bh_r - \cff Y_{\beta}r^{2H}\right\} \middle \vert  \sigma(Y_\beta)\right]\right]
    r^{\Hm} \D r. \nonumber\\
    &= \xi_0 \eta \cf \frac{\sqrt{\pi}}{2\Gamma(1+\half \beta)}\frac{\Delta^{H + \half}}{H + \half} ,
\end{align}
%\adri{Change: recalled $\Gamma(1+\half)=\frac{\sqrt{\pi}}{2}$ and explained why the conditional expectation in the second-to-last line equals one (Reviewer 2, Comment 26).}
\blue{where $\VIX_0 = \sqrt{\xi_0}$; the conditional expectation in the second-to-last line equals one since $\Bh_r \sim N\left(0, \cf^2 r^{2H}/(2H)\right)$ and using~\eqref{eq:MWrightMom}. 
Similar calculations yield}
\begin{align}
    J_2 &= \xi_0 \eta^2 \cf^2 \frac{\Delta^{2H}}{\Gamma(1 + \beta)2H}, \label{eq:J2}\\
    J_3(T) &= \xi_0 \eta^3 \cf^3 \frac{3\sqrt{\pi}}{4\Gamma(1+ \frac{3}{2}\beta)(3H - \half)} \left(\left(T + \Delta \right)^{3H - \frac{1}{2}} - T^{3H - \frac{1}{2}}\right).\label{eq:J3}
\end{align}
To conclude, we have that if $H < \frac16$, then 
\begin{equation*}
    \lim_{T \downarrow 0} T^{\half -3H} J_3
    = - \frac{3\xi_0 \eta^3 \cf^3\sqrt{\pi}}{4\Gamma(1+ \frac{3}{2}\beta)(3H - \half)}.
\end{equation*}
\vspace*{-\baselineskip}
\end{proof}

%\adri{Change: added the $H\in(1/6,1/2)$ analogue of Remark~3 in~\cite{Jacquier2021RoughOptions} for the VIX curvature (Reviewer 2, Comment 25).}
\blue{
\begin{remark}
For $H\in(\tfrac16,\tfrac12)$ the rescaling exponent of
Proposition~\ref{prop:gBergomi_VIX_asym} vanishes and the ATM VIX curvature converges to a
finite limit, the analogue of~\cite[Remark~3]{Jacquier2021RoughOptions}; the many terms that
are subdominant when $H<\tfrac16$ now all contribute at the same order. Since the VIX
Clark--Ocone representation involves only the single driver $Z^2=B$, the multi-factor limit
of~\cite[Remark~3]{Jacquier2021RoughOptions} collapses to one index. Writing
$a:=\Delta\VIX_0^2=\Delta\xi_0$, recalling $J_1,J_2,J_3$ from~\eqref{eq:J1}--\eqref{eq:J3} it reads informally (that is,
without proof, as in~\cite{Jacquier2021RoughOptions})
$$
\lim_{T\downarrow0}\Cc_T
= \lim_{T\downarrow0}\Ss_T
- \frac{15a}{J_1^3}\Big(J_2-\tfrac{J_1^2}{a}\Big)^2
+ \frac{12a}{J_1^4}\Big(J_2-\tfrac{J_1^2}{a}\Big)^2
+ \frac{9J_1}{2a}
- \frac{18 J_2}{J_1}
+ \frac{2a}{3J_1^2}J_3(0),
$$
where $J_3(0)=\dfrac{3\sqrt\pi\xi_0\eta^3\cf^3\Delta^{3H-\half}}{4\Gamma(1+\frac32\beta)(3H-\half)}$
is the (now finite) $T\downarrow0$ value of~\eqref{eq:J3}. In contrast to the case
$H<\tfrac16$, the curvature is bounded and determined jointly by the first three Malliavin
orders rather than by $J_3$ alone. The borderline $H=\tfrac16$ is the transitional case, where
$J_3(T)\sim\log(1+\Delta/T)$ diverges logarithmically. We caution that, unlike the Gaussian
rough Bergomi setting, the grey randomisation enters each $J_k$ through a distinct
Mittag-Leffler factor ($\Gamma(1+\tfrac12\beta)$, $\Gamma(1+\beta)$, $\Gamma(1+\tfrac32\beta)$),
so the covariance combinations $J_2-J_1^2/a$ do not simplify as they do at $\beta=1$; a fully
rigorous limit would require conditioning on $Y_\beta$. Setting $\beta=1$ recovers the rough
Bergomi values of~\cite[Remark~3]{Jacquier2021RoughOptions}.
\end{remark}
}
\begin{remark}
When $\beta=1$, then
$J_1 = \xi_0 \eta \cf \frac{\Delta^{\Hp}}{\Hp}$,
$J_2 = \xi_0 \eta^2 \cf^2 \frac{\Delta^{2H}}{2H}$ and, 
for $H \in\left(0, \half\right)$,
$$
\lim _{T \downarrow 0} \Ss_T =
\frac{J_2}{2J_1} - \frac{J_1}{2\Delta\VIX_0^2}
 = \frac{\eta \Delta^{\Hm}}{2\Gamma(\Hp)}
 \left(\frac{\Hp }{2H} - \frac{1}{\Hp}\right).
$$
%\adri{Change: added note on $H\uparrow\half$ limit of the skew (Reviewer 2, Comment 27).}
\blue{Note that as $H\uparrow\tfrac{1}{2}$, the skew vanishes.}
\end{remark}

%\adri{Change: expanded "In particular" to note the filtration gap ($\tilde{\mathcal{F}}^B_u \supseteq \mathcal{F}^B_u$, VIX uses $\mathcal{F}^B$) and added Lyapunov step to proof for $p \leq 2$ (Reviewer 2, Comment 28).}
\begin{lemma}\label{lem:AssOneAsy}
    In the gBergomi model~\eqref{eq:gBergomi_dynamics} with $0 \leq T_1<T_2$,
    $$
    \EE\left[\sup _{u \leq T_1}\left(\EE\left[\frac{1}{T_2-T_1} \int_{T_1}^{T_2} V_r \D r\middle\vert \FfB_u\right]\right)^{-p}\right]
    $$
    is finite for all $p>1$. In particular, $1 / \Ffr^{\VIX}$ is dominated in $L^p$.
\end{lemma}
\begin{proof}
%\adri{Change: expanded proof for clarity -- added labels for Jensen step, equality step, Jia bound, sign step, Gaussian supremum, and M-Wright MGF. Reviewer 1, Comment 17.}
Similarly to~\cite[Lemma~9]{Jacquier2021RoughOptions},
by the exp-log identity and Jensen's inequality,
\begin{align*}
 &     \EE\left[\frac{1}{T_2-T_1} \int_{T_1}^{T_2} V_r \D r  \middle\vert  \FfB_u\right]^{-p}\\
 & = \blue{\exp \left\{-p \log \EE\left[\frac{1}{T_2-T_1} \int_{T_1}^{T_2} V_r \D r  \middle\vert  \FfB_u\right]\right\}} \\
    &\leq \exp \left\{-\frac{p}{T_2-T_1} \int_{T_1}^{T_2} \log \left(\EE\left[V_r\middle\vert \FfB_u\right]\right) \D r\right\} \\
    &= \exp \left\{-\frac{p}{T_2-T_1} \int_{T_1}^{T_2} \log \left(\frac{\xi_0}{\Ef_{\beta}\left(\frac{\eta^2 r^{2H}}{2}\right)}\EE\left[\exp\left\{\eta\sqrt{Y_{\beta}} \Ba_r\right\}\middle\vert \FfB_u\right]\right) \D r\right\} \\
    & \leq \exp \Bigg\{-\frac{p}{T_2-T_1} \int_{T_1}^{T_2} \left[\log(\xi_0) + C\left(\frac{\eta^2 r^{2H}}{2}\right) + \log \left(\EE\left[\exp\left\{\eta\sqrt{Y_{\beta}} \Ba_r\right\}\middle\vert \FfB_u\right]\right)\right] \D r\Bigg\},
    \end{align*}
    %\adri{Change: (1) restored missing conditional expectation in second line of align; (2) corrected $C>0$ to $C<0$ — from Theorem~9 of~\cite{Jia2019SomeFunction}, $\log\mathcal{E}_\beta(x)\leq\log C'+x$, so $-\log\mathcal{E}_\beta(x)\geq -\log C'-x$ and the coefficient of $x$ is $-1<0$ (Reviewer 2, Comment 29).}
    %\adri{Change: noted that the constant $C$ from Jia Theorem~9 depends on $\beta$ and is not uniform in $\beta$, but is finite for each fixed $\beta\in(0,1]$. Reviewer~1, Comment~16.}
    \blue{The first inequality applies Jensen's inequality to the concave function $\log$ over the uniform measure on $[T_1,T_2]$: since $r\mapsto\EE[V_r\mid\FfB_u]$ is the integrand, Jensen gives 
    $$\log\left(\frac{1}{T_2-T_1}\int_{T_1}^{T_2}\EE[V_r\mid\FfB_u]\D r\right)\geq \frac{1}{T_2-T_1}\int_{T_1}^{T_2}\log\EE[V_r\mid\FfB_u]\D r.
    $$
    The equality uses that $\xi_0(r)/\Ef_\beta\bigl(\frac{\eta^2 r^{2H}}{2}\bigr)$ is deterministic and hence factors out of the conditional expectation $\EE[\cdot\mid\FfB_u]$}
    for some \blue{$C = C(\beta) < 0$}, where the estimate~\cite[Theorem~9]{Jia2019SomeFunction} was used in the last inequality. \blue{Specifically, $\Ef_\beta(x)\leq C'(\beta)e^x$ for a finite constant $C'(\beta)>0$ so that $-\log\Ef_\beta(x)\geq -\log C'(\beta) - x$. 
    The constant $C(\beta)$ is not uniform in~$\beta$, but is finite for each fixed $\beta\in(0,1]$; since~$\beta$ is a fixed model parameter throughout, the bound is valid.} Since the first two terms in the integral are clearly finite, we now turn our attention to the conditional expectation. %\adri{Change: updated citation from Theorem~3.1 to Remark~3.1(b) (better fit) and replaced $Y_\beta = y$ conditioning with $Y_\beta$ as a random variable throughout, conditioning on $\FfB_u$ directly (Reviewer 2, Comment 30).} By~\blue{\cite[Remark 3.1(b)]{Fink2013}} we have
\begin{align*}   &\int_{T_1}^{T_2}\log\left(\EE\left[\exp\blue{(\eta\sqrt{Y_\beta} B_r^{H})}\middle\vert \blue{\FfB_u}\right]\right) \D r \\
    & = \int_{T_1}^{T_2} \left[\blue{\eta\sqrt{Y_\beta}}\left(B_u^{H} + \int_0^u \Psi(u,r,v)\D B_v^{H}\right) + \frac{\blue{\eta^2 Y_\beta}}{2}\left(|r-u|^{2H} - \EE\left[\left|\int_0^u \Psi(u,r,v)\D B_v^{H}\right|^2\right]\right)\right] \D r
\end{align*}
and
\begin{equation}\label{eq:Psi}
\Psi(s,t,u) = \frac{\sin \pi \Hm}{\pi} u^{-\Hm}(s-u)^{-\Hm} \int_s^t \frac{z^{\Hm}(z-s)^{\Hm}}{z-u} \D z.
\end{equation}
The second term can be written as an integral of the kernel function
$K(u,r,v)\coloneqq u^{\Hm} \int_0^r z^{\Hm} (z - v)^{\Hm - 1}\D z$
with respect to the related Brownian motion~\cite[Equation~(7.2)]{Pipiras2001}, so that by Fubini,
\begin{equation}\label{eq:integralfBmFubini}
\int_{T_1}^{T_2} \int_0^u \Psi(u,r,v)\D B_v^{H} \D r = \int_0^u \int_{T_1}^{T_2} K(u,r,v) \D r \D B_v \eqqcolon \overline{B}_u.
\end{equation}
Since the kernel~$K$ is integrable~\cite[Theorem~4.2]{Pipiras2001}, $\overline{B}_u$ is a centred Gaussian process with square-integrable paths. \blue{Applying the tower property to condition on $\sigma(Y_\beta)$, and using $-\overline{B}_u \leq |\overline{B}_u| \leq \sup_{t\in[0,T]}|\overline{B}_t|$,}
\begin{align*}
&\EE\left[\sup _{u \leq T_1} \exp \left(-\frac{p\eta}{T_2-T_1} Y_\beta \int_{T_1}^{T_2} \int_0^u \Psi(u,r,v)\D B_v^{H} \D r\right)\right] \\
& \hspace{2cm} = \EE\left[ \EE\left[\sup _{u \leq T_1} \exp \left(-\frac{p\eta}{T_2-T_1} Y_\beta \overline B_u\right)\middle\vert \sigma(Y_\beta)\right] \right] \\
& \hspace{2cm} \leq \EE\left[ \EE\left[ \exp \left(\blue{\frac{p\eta}{T_2-T_1}} Y_\beta \sup_{t\in[0,T]}|\overline B_t|\right)\middle\vert \sigma(Y_\beta)\right] \right] \\
& \hspace{2cm} \leq \EE\left[ \EE\left[ \exp \left(\blue{\frac{p\eta}{T_2-T_1}} Y_\beta M\right)\middle\vert \sigma(Y_\beta)\right] \right]
\end{align*}
%\adri{Change: corrected sign error — $\sup_u\exp(-cY_\beta\bar{B}_u)\leq\exp(+cY_\beta\sup|\bar{B}_t|)$ since $-\bar{B}_u\leq|\bar{B}_u|\leq\sup|\bar{B}_t|$; minus signs removed from lines 907, 908 and from the argument of $\mathcal{E}_\beta$ in the conclusion (Reviewer 2, Comment 31).}
\blue{for some finite deterministic constant $M>0$, where the last step uses~\cite[Lemma~8]{Jacquier2021RoughOptions}: since $\overline{B}$ is a centred Gaussian process on the compact interval $[0,T_1]$, its supremum has sub-Gaussian tails, giving $\EE[\exp(c\sup_{u\leq T_1}|\overline{B}_u|)]\leq e^{cM}$ for all $c>0$.}
Applying the tower property and the MGF identity $\EE[e^{tY_\beta}] = \Ee_\beta(t)$ (from~\eqref{eq:LaplaceM-Wright}) we conclude that
$$
\EE\left[\sup _{u \leq T_1} \exp \left(-\frac{p\eta}{T_2-T_1} Y_\beta \int_{T_1}^{T_2} \int_0^u \Psi(u,r,v)\D B_v^{H} \D r\right)\right] \leq \Ee_\beta\left(\blue{\frac{\eta p}{T_2-T_1} M}\right) < \infty,
$$
\blue{where finiteness holds for each fixed $\beta\in(0,1]$ since $\mathcal{E}_\beta$ is an entire function.}
Clearly $\int_{T_1}^{T_2}|r-u|^{2H} \D r$ is bounded from below for
$u \in [0,T_1]$ and the related term finite. 
For the last term,  the same representation as above and It\^o's isometry give
\begin{align*}
    \int_{T_1}^{T_2}\EE\left[\left|\int_0^u \Psi(u,r,v)\D B_v^{H}\right|^2\right] \D r 
    &= \int_{T_1}^{T_2} \EE\left[\left|\int_0^u K(u,r,v) \D B_v\right|^2\right] \D r \\
    &= \int_{T_1}^{T_2} \int_{0}^{u}|K(u,r,v)|^2\D v \D r,
\end{align*}
where~$K$ is twice locally integrable with respect to~$v$ by~\cite[Theorem~4.2]{Pipiras2001} for every $r\geq 0$, hence the outer integral over the compact $[T_1, T_2]$ is finite.
\blue{To conclude, since $\VIX^2_{T_1} = \frac{1}{T_2 - T_1}\int_{T_1}^{T_2}\EE[V_r \mid \Ff^B_{T_1}]\D r$ and $\Ff^B_{T_1} \subseteq \FfB_{T_1}$, the convexity of $x^{-p}$ and the tower property give
$$
\EE\bigl[(\VIX^2_{T_1})^{-p}\bigr] \leq \EE\left[\left(\EE\left[\tfrac{1}{T_2-T_1}\int_{T_1}^{T_2} V_r\D r\middle\vert\FfB_{T_1}\right]\right)^{-p}\right] < \infty, \quad p > 1.
$$
For $p \leq 2$, we cannot directly write $\EE[(\VIX)^{-p}] = \EE[(\VIX^2)^{-\frac{p}{2}}]$ and apply the above since $\frac{p}{2} \leq 1$. Instead, choose any $q > 1$ and apply Lyapunov's inequality (Jensen applied to the concave map $t \mapsto t^{\frac{p}{2q}}$ with $\frac{p}{2q} < 1$):
$$
\EE\left[\left(\VIX^2\right)^{-\frac{p}{2}}\right]
\leq \EE\left[\left(\VIX^2\right)^{-q}\right]^{\frac{p}{2q}} < \infty.
$$
Hence $1/\Ffr^{\VIX} \in L^p$ for all $p > 1$, completing the proof.}
\end{proof}

\begin{lemma}\label{lem:AssThreeAsy}
    For any $p>1$, \blue{$\EE\left[h_s^{-p}\right]$} is uniformly bounded in $s\in [0,T]$.
\end{lemma}
\begin{proof}
    The proof follows the argument of~\cite[Section~6]{Jacquier2021RoughOptions}. Since $\Dr_u V_r$ is positive for all $u\leq r$ almost surely and since $\VIX$ and $1/\VIX$ are dominated by $R\in L^p$ for all $p\geq 1$, then
    \begin{equation*}
        \ff_u \coloneqq \ff_u(\VIX_T) = \EE\left[ \frac{\int_T^{T+\Delta}\Dr_u V_r \D r}{2 \Delta \VIX_T}\middle\vert \FfB_u\right] \geq \EE\left[\frac{\int_T^{T+\Delta}\Dr_u V_r \D r}{2 \Delta R}\middle\vert \FfB_u\right],
    \end{equation*}
    almost surely.
    Since $1/\Ffr \leq R$ by Lemma~\ref{lem:AssOneAsy} Jensen's inequality implies
    \begin{align*}
        \blue{h_s^{-2}} = \frac{T-s}{\int_s^T|\phi_u|^2 \D u} \leq \frac{R^2 (T-s)}{\int_s^T |\ff_u|^2 \D u} \leq \left| \frac{R \blue{(T-s)}}{\int_s^T \ff_u \D u} \right|^2 \leq 4R^2 \left( \frac{\int_s^T \int_T^{T+\Delta} \EE\left[\frac{\Dr_u V_r}{R}\middle\vert \FfB_u\right] \D r \D u}{\blue{\Delta(T-s)}}\right)^{-2}.
    \end{align*}
    %\adri{Change: corrected normalisation from $1/(\Delta\sqrt{T-s})$ to $1/(\Delta(T-s))$ throughout (Jensen and Cauchy--Schwarz steps); no extra $(T-s)^{-p/2}$ factor arises. Reviewer~2, Comment~32.}
    \sloppy Next, by using the exp-log identity together with Jensen's, Cauchy-Schwarz and ${\exp\left(p\EE_u[\log R]\right) \leq \EE_u[R^p]}$ inequalities we have
    \begin{align*}
        &\EE\left[\left(\blue{\frac{1}{\Delta(T-s)}} \int_s^T \int_T^{T+\Delta} \EE\left[\frac{\Dr_u V_r}{R}\middle\vert\FfB_u\right] \D r \D u\right)^{-p}\right] \\
        &= \EE\left[\exp\left\{-p\log\left(\blue{\frac{1}{\Delta(T-s)}} \int_s^T \int_T^{T+\Delta} \EE\left[\frac{\Dr_u V_r}{R}\middle\vert\FfB_u\right] \D r \D u\right)\right\}\right] \\
        &\leq \EE\left[\exp\left\{-\blue{\frac{p}{\Delta(T-s)}}\int_s^T \int_T^{T+\Delta} \EE\left[\log\Dr_u V_r - \log R\middle\vert\FfB_u\right] \D r \D u\right\}\right] \\
        &\leq \left(\EE\left[\exp\left\{-\blue{\frac{2p}{\Delta(T-s)}}\int_s^T \int_T^{T+\Delta} \EE\left[\log\Dr_u V_r\middle\vert\FfB_u\right] \D r \D u\right\}\right]\right)^{\half} \sqrt{\EE\left[R^{2p}\right]}.
    \end{align*}
Focusing on the conditional expectation, by computations in~\eqref{eq:malliavin_gBergomi},
\begin{align*}
    \EE\left[\log\Dr_u V_r\middle\vert\FfB_u\right] &= \EE\left[\log\left(\eta \cf \sqrt{Y_\beta} V_r (r-u)^{\Hm}\right)\middle\vert\FfB_u\right] \\
    &= \log\left(\eta \cf^2 \sqrt{Y_\beta} (r-u)^{\Hm}\right) + \EE\left[\log V_r\middle\vert\FfB_u\right],
\end{align*}
the first term is uniformly bounded by the same computations as in~\cite[Section~6]{Jacquier2021RoughOptions} treating $Y_\beta$ as a constant and relying on the existence of MGF of M-Wright distribution. As for the expectation of the log-volatility term,
\begin{align*}
    \EE\left[\log V_r\middle\vert\FfB_u\right] = \log\xi_0 - \log\Ef_\beta\left(\frac{\eta^2 r^{2H}}{2}\right) + \eta\EE\left[\sqrt{Y_\beta} \Ba_r\middle\vert\FfB_u\right],
\end{align*}
the double integral over a compact corresponding to the second term is uniformly bounded, since~\cite[Theorem~9]{Jia2019SomeFunction} gives the bound~$1\leq\Ef_\beta(x)\leq C\E^{x}$ for all $x\geq 0$ and some constant~$C>0$. Then by~\cite[Theorem~7.1]{Pipiras2001}, the second term equals to
\begin{equation*}
    \EE\left[\sqrt{Y_\beta} \Ba_r\middle\vert\Ff^B_u, Y_\beta\right] = \sqrt{Y_\beta}\left(\Ba_u + \int_0^u\Psi(u,r,v)\D \Ba_v\right),
\end{equation*}
whose corresponding integrals over a compact are again finite by~\eqref{eq:integralfBmFubini} and the ensuing arguments in Lemma~\ref{lem:AssOneAsy}.
\end{proof}

Although not immediately obvious from the Proposition~\ref{prop:gBergomi_VIX_asym}, numerical experiments show that the short-time ATM skew is positive for all choices of parameters, which is indeed what we observe, an upward-sloping VIX smile.
Figure~\ref{fig:vixasymptotics} showcases the behaviour in~$\beta$ of the asymptotic results from Proposition~\ref{prop:gBergomi_VIX_asym}.

\begin{figure}[htb]
\centering
    \subfloat{\includegraphics[width=0.45\linewidth]{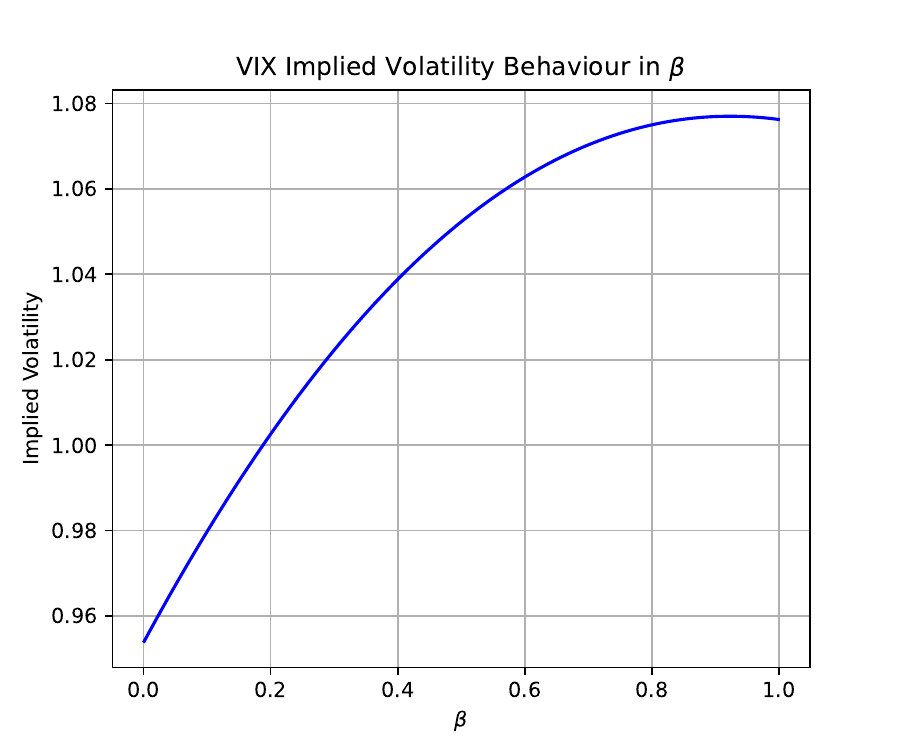}}
    \subfloat{\includegraphics[width=0.45\linewidth]{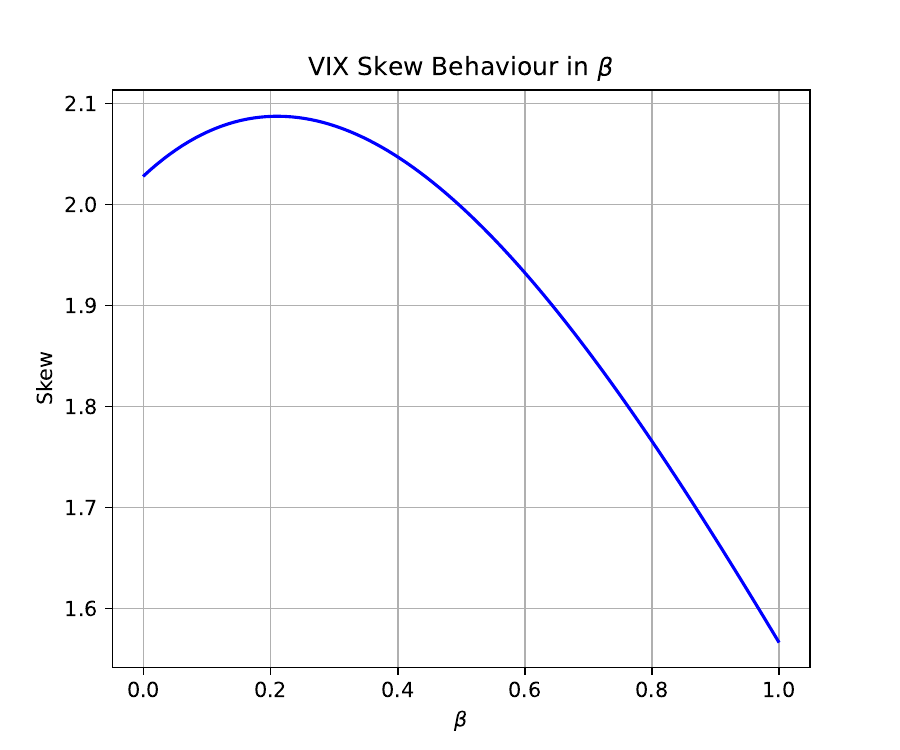}}
    \newline\vspace*{-0.75\baselineskip}
    \includegraphics[width=0.45\linewidth]{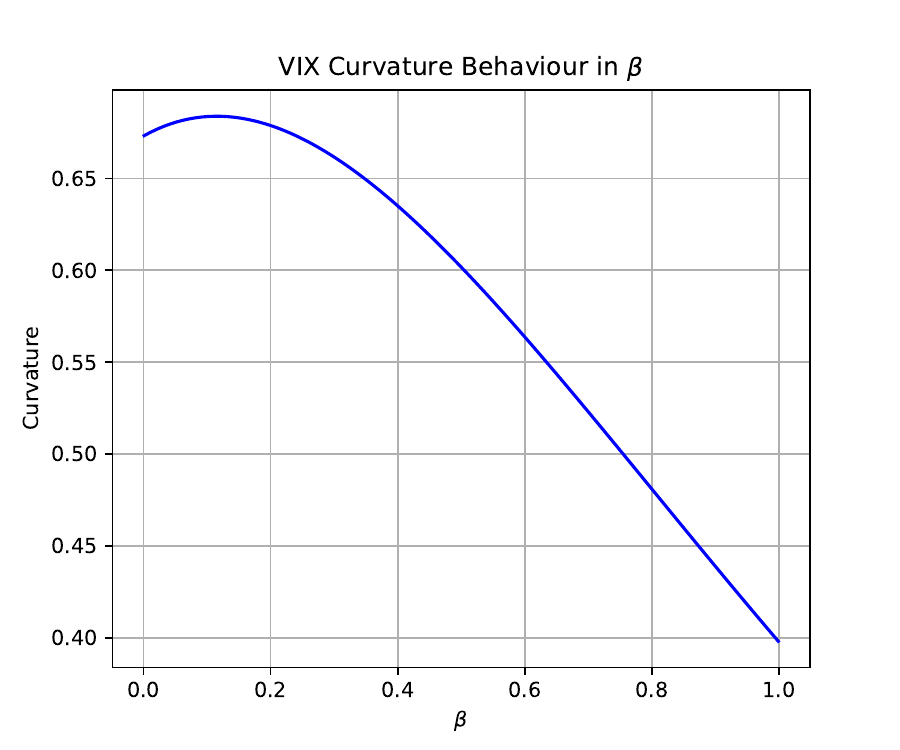}
    \caption{At-the-money VIX level, skew, curvature asymptotics with
    $(H,\eta, \xi_0) = (0.065, 0.65, 0.235^2)$, taken from the calibration results.}
\label{fig:vixasymptotics}
\end{figure}

\subsection{SPX asymptotics}\label{sec:SPX_asymptotics}
To differentiate from the VIX, we slightly modify the notations, denoting the ATM implied volatility level by~$\widehat \Ii_T$ and the skew by~$\widehat\Ss_T$. 
The general set-up above still applies in the case of SPX, however now with two sources of noise.
The following proposition summarises this, and is numerically illustrated in Figure~\ref{fig:spxskew}.

\begin{proposition}\label{prop:gBergomi_SPX_asym}
\blue{For $H \in (0, \half)$,} the following small~maturity behaviours hold:
$$
\lim _{T \downarrow 0} \widehat{\Ii}_T = \sqrt{\xi_0}
\qquad\text{and}\qquad
\lim_{T\downarrow 0} \frac{\widehat\Ss_T}{T^{H + \frac{3}{2}}} =  \frac{\rho\eta \cf\sqrt{\pi}}{(2H + 1)(2H + 3)\Gamma(1+\half \beta)}.
$$
\end{proposition}
\begin{proof}
The proof of the proposition relies on
\cite[Proposition~4]{Jacquier2021RoughOptions}, for which we need to check the following assumptions:
    There exists $H \in\left(0, \half\right)$ and a random variable $R$ such that, for all $0 \leq s \leq u$, % $j\in\{1,2\}$ 
    and $p \geq 1, R \in L^p$,
    \begin{enumerate}[(i)]
        \item $V_s \leq R$ almost surely;
        \item $\Dr_s V_u \leq R(u-s)^{\Hm}$ almost surely;
        \item $\sup _{s \leq T} \blue{\EE\left[h_s^{-p}\right]}<\infty$;
        \item $\limsup _{T \downarrow 0} \EE[(\sqrt{V_T / V_0}-1)^2]=0$.
    \end{enumerate}
Under these assumptions, Proposition~4 in~\cite{Jacquier2021RoughOptions} states that the short-time limit of the implied volatility is given as in the proposition and the short-time skew reads (since $\Dr^W_u V_s = 0$, the Brownian motion driving the stock does not play a role~\cite[Section~5]{Jacquier2021RoughOptions})
$$
\lim _{T \downarrow 0} \frac{\widehat{\Ss}_T}{T^{H + \frac{3}{2}}}
 = \frac{\rho}{2 V_0}
     \lim _{T \downarrow 0} \frac{\int_0^T \int_s^T \EE\left[\Dr_s V_u\right] \D u \D s}{T^{H + \frac{3}{2}}}.
$$
The proof of Proposition~\ref{prop:gBergomi_VIX_asym} establishes that the assumptions (i)-(iii) are satisfied. 
%\adri{Change: replaced invalid reverse Fatou argument with dominated convergence. Reviewer~1, Comment~12.}
As for assumption~(iv), since $\{V_t\}_{t\geq 0}$ has almost surely continuous paths, the ratio~$V_t/V_0$ converges to~$1$ almost surely, so $(\sqrt{V_T/V_0}-1)^2\to 0$ almost surely. \blue{By assumption~(i), $V_T \leq R$ almost surely, so $(\sqrt{V_T/V_0}-1)^2 \leq 2(R/V_0+1)$, which is integrable since $R\in L^1$. Dominated convergence then gives~(iv).}
For the ATM skew the same calculations as in Section~\ref{sec:VIX_asymptotics} yield
$$
\EE\left[\Dr_s V_u\right]
= \EE\left[\eta \cf \sqrt{Y_\beta} V_u (u-s)^{\Hm}\right]
= \frac{\xi_0 \eta \cf\sqrt{\pi}}{2\Gamma(1+\half \beta)}(u-s)^{\Hm}.
$$
The result then follows
from the computation of the double integral
$$
\int_0^T \int_s^T\EE\left[\Dr_s V_u\right]\D u \D s
= \frac{\xi_0 \eta \cf\sqrt{\pi}}{\Gamma(1+\frac{\beta}{2})}\frac{T^{H + \frac{3}{2}}}{(H + \half)(2H + 3)}.
$$
\end{proof}

\begin{figure}[ht]
\centering
    \includegraphics[width=0.5\linewidth]{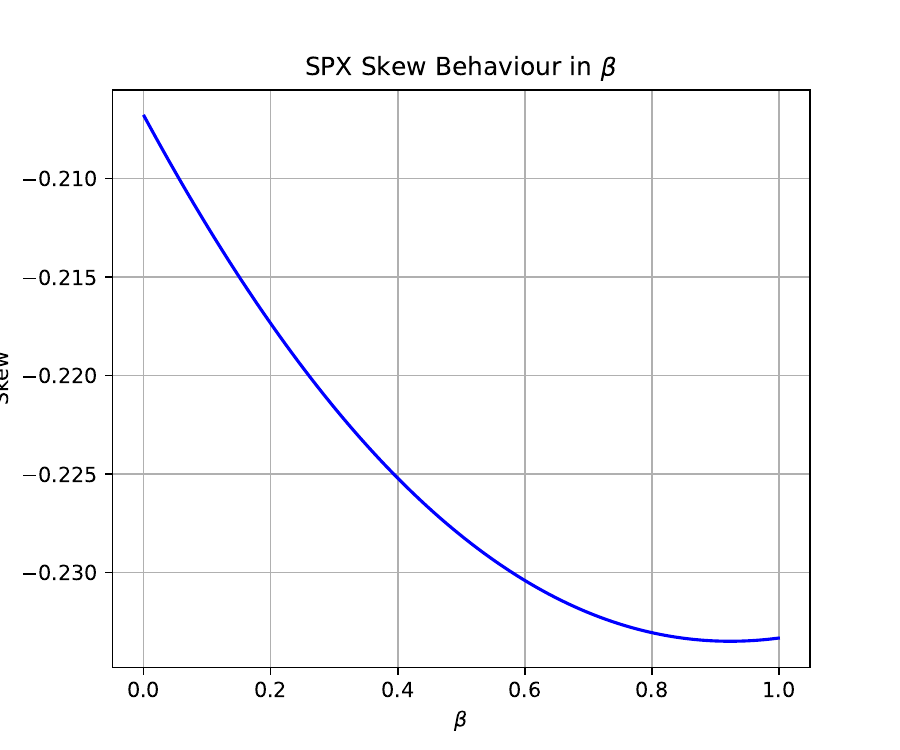}
    \caption{Short-term SPX skew with~$(H, \eta, \rho) = (0.065, 0.65, -1)$.}
\label{fig:spxskew}
\end{figure}

\blue{
\begin{remark}\label{rem:spx_asym_interp}
Proposition~\ref{prop:gBergomi_SPX_asym} separates the parameter roles in the
short-maturity SPX smile. The leading-order ATM level is
$\lim_{T\downarrow0}\widehat{\Ii}_T=\sqrt{\xi_0}$, independent of $\beta$, $\rho$ and
$\eta$: at first order the SPX level is pinned by the forward variance alone. The ATM skew
scales as $T^{H+\frac32}$ with prefactor
$\frac{\rho\eta\cf\sqrt{\pi}}{(2H+1)(2H+3)\Gamma(1+\half\beta)}$, so its sign and
dominant magnitude are set by the spot--volatility correlation $\rho$ (the skew vanishes as
$\rho\to0$), while the roughness $H$ governs the rate of the power-law blow-up as
$T\downarrow0$. The grey parameter $\beta$ enters the SPX skew only through the factor
$1/\Gamma(1+\half\beta)$, which varies monotonically but mildly over $\beta\in(0,1]$
(by roughly $11\%$, from $1$ at $\beta\to0$ to $2/\sqrt{\pi}\approx1.128$ at $\beta=1$):
larger $\beta$ yields a marginally steeper skew, the rough Bergomi limit $\beta=1$ being the
steepest. Consequently $\beta$ is
nearly invisible to the short-dated SPX surface --- it leaves the ATM level untouched and
only weakly modulates the skew --- and cannot be identified from SPX data alone; its identification rests on the VIX smile.
\end{remark}
}
%%%%%%%%%%%%%%%%%%%%%%%%%%%%%%%%%%%%%%%%%%%%
%%%%%%%%%%%%%%%%%%%%%%%%%%%%%%%%%%%%%%%%%%%%

\section{Joint Calibration}\label{sec:calibration}
We now focus on calibrating the gBergomi model to market data\footnote{SPX and VIX option data were obtained from OptionsDX (\url{https://www.optionsdx.com}).} as of 17/10/2023. We proceed with a two-step calibration approach where in both cases we minimise the MSE of the market and model implied volatilities. First, we calibrate the VIX smile over $(H,\beta,\eta)$. We then calibrate the SPX smile over $(\beta,\eta,\rho)$, while keeping the roughness parameter $H$ fixed at its VIX-calibrated value.

The VIX smile calibration yields in the following calibrated values
$$
(H^*, \beta^*, \eta^*) = (0.065, 0.037, 0.65).
$$
The SPX smile calibration,with $H$ fixed at $0.065$, yields
$$
(\beta^*, \eta^*,\rho^*) = (1.0, 1.09, -1).
$$
Thus, the SPX calibration drives the model to the rough Bergomi limit $\beta\to1$, while the VIX calibration requires a strongly grey regime with $\beta$ close to zero. The calibrated SPX volatility-of-volatility parameter is also approximately $68\%$ larger than the VIX-calibrated value. A related discrepancy between the SPX and VIX volatility-of-volatility parameters was observed in~\cite{jacquier2018vix}, where it was interpreted as evidence of an apparent inconsistency between the two markets, and potentially as a symptom of SPX--VIX dislocation. In our setting, the discrepancy is even more structural: it appears not only in the level of volatility of volatility, but also in the grey parameter $\beta$ itself.
\begin{figure}[ht]
    \centering
    \includegraphics[width=0.75\linewidth]{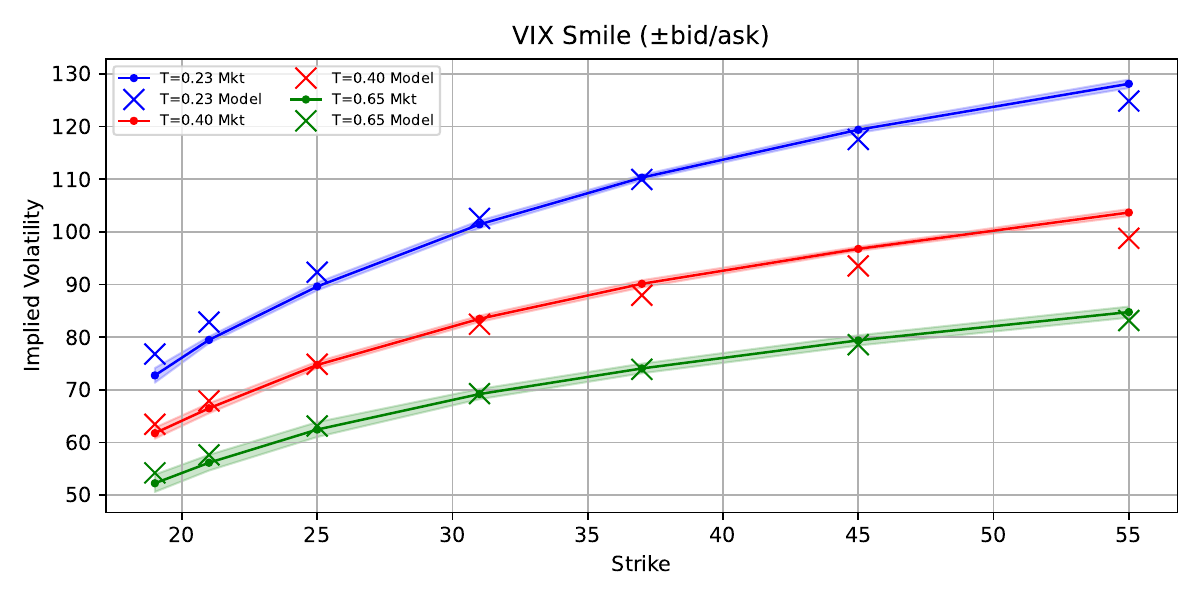}
    \caption{VIX smile on 25/10/2023. 
    $(H, \beta, \eta) = (0.065, 0.037, 0.65)$.}
    \label{fig:vixsmile}
\end{figure}

\begin{figure}[ht]
    \centering
    \includegraphics[width=0.75\linewidth]{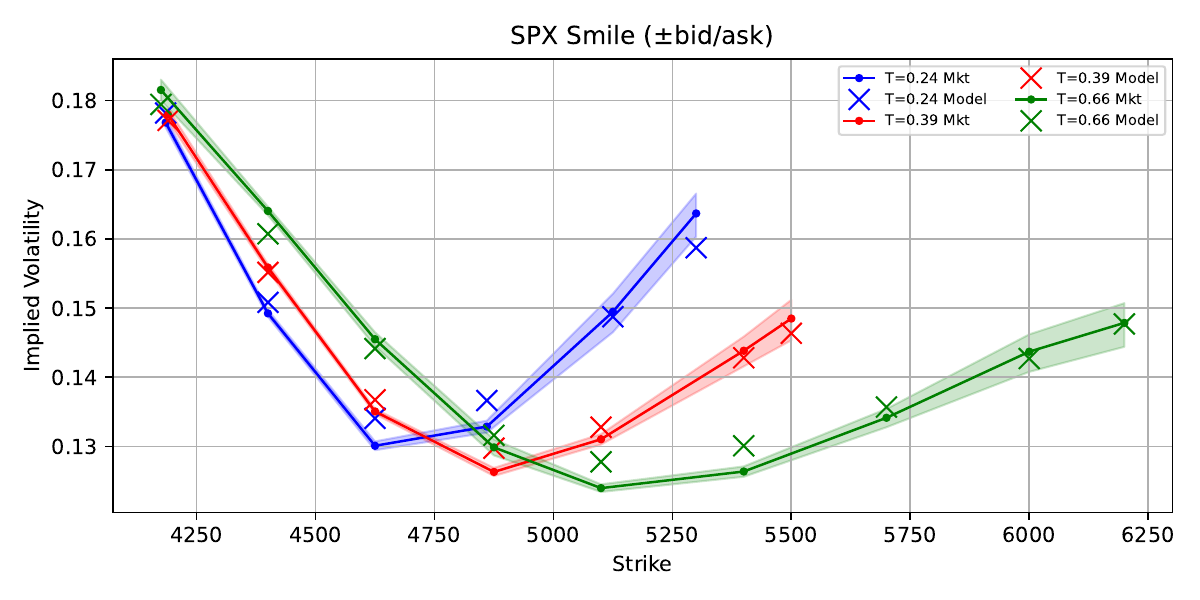}
    \caption{SPX smile on 25/10/2023.
    $(H, \beta, \eta, \rho) = (0.065, 1, 1.09, -1)$.}
    \label{fig:spxsmile}
\end{figure}

%\adri{Change: added a comment on calibration computation time (Reviewer 2, Comment 37).}
\blue{
\begin{remark}
Both calibrations use differential evolution~(\texttt{differential\_evolution} from \texttt{scipy.optimize}, single-threaded) with \texttt{maxiter}~$=30$ and \texttt{popsize}~$=15$, on a 1.6GHz dual-core Intel Core i5 with 8GB of memory. With $d=3$ free parameters per stage, the search budget is at most $\texttt{popsize}\times d\times(\texttt{maxiter}+1)\approx 1.4\times10^{3}$ objective evaluations, each a full Monte Carlo pricing at $N=10^{5}$ paths and $252$ steps per year. At
roughly $4.8$s and $2.9$s per VIX- and SPX-smile evaluation respectively, the VIX calibration over $(H,\beta,\eta)$ takes $\approx 1.9$ hours and the subsequent SPX calibration over $(\beta,\eta,\rho)$ up to $\approx1.1$ hours.
\end{remark}
}

\blue{ 
The SPX smile is primarily a crash-skew market. It reflects the fact that market participants pay for volatility to rise when the index falls, and therefore it requires a strong negative spot--volatility correlation. In contrast, the VIX smile is primarily a volatility-convexity market. It prices optionality on future variance itself becoming extremely large, whether or not this variance shock is explained one-for-one by the contemporaneous equity move. In the one-factor gBergomi model, the same grey parameter $\beta$ is forced to serve both purposes. When $\beta\to1$, the grey randomisation disappears and the model reduces to the rough Bergomi limit. This is the regime preferred by the SPX surface: the SPX smile wants leverage-driven skew, generated through a large negative correlation $\rho$, and does not require an additional independent variance-tail randomisation. Conversely, when $\beta$ is small, the M-Wright random variable introduces a heavy-tailed volatility-of-volatility mixture. This is precisely the mechanism preferred by the VIX surface, since VIX options are highly sensitive to convexity in future variance. Thus, the economic conflict is that SPX wants the grey tail switched off because it interferes with the leverage-skew mechanism, whereas VIX wants the grey tail switched on because it generates the variance-tail convexity priced by VIX calls. In market terms, SPX options are primarily protection against equity downside and crash beta, while VIX options are closer to pure volatility-spike convexity. These are related risks, but they are not identical risks. A single grey factor is therefore asked to play two incompatible roles. This suggests a natural interpretable extension. Rather than assigning separate models to SPX and VIX, one may keep a single common variance process but separate the leverage-skew channel from the variance-tail-convexity channel. One possible specification is 
$$
\log\left(\frac{V_t}{\xi_0(t)}\right)
= \eta_{\mathrm{lev}} c_H X_t + \eta_{\mathrm{tail}} c_H \sqrt{Y_\beta} Z_t - \Lambda(t),
$$ 
where 
$X_t := \int_0^t (t-s)^{H-\half}\D  B_s$ and 
$Z_t := \int_0^t (t-s)^{H-\half}\D  B_s^\perp$.
The normalisation term~$\Lambda(t)$ is chosen such that that $\EE[V_t]=\xi_0(t).$ The equity index is driven by $$ \frac{\D  S_t}{S_t} = \sqrt{V_t}\D  W_t, $$ with 
$\D\langle W,B\rangle_t = \rho\D t$
and
$\D\langle W,B^\perp\rangle_t=0$.
In this specification, both SPX and VIX options are written on the same variance process $V_t$. The model does not assign one factor to SPX and another factor to VIX in an instrument-specific way. Instead, the common variance process contains two economically distinct channels. The loading $\eta_{\mathrm{lev}}$ controls the spot-correlated leverage-skew channel, while $ \eta_{\mathrm{tail}}$ and $\beta$ control the grey variance-tail-convexity channel. The parameter $\rho$ controls the strength of the equity crash--volatility feedback, while $\beta$ controls the heaviness of the independent volatility-of-volatility tail. This extension separates two risks that traders already distinguish in practice. The first is crash-directional volatility risk: when the index sells off, volatility rises. This is the risk primarily embedded in the SPX skew, and it is naturally represented by the spot-correlated rough factor $X_t$. The second is pure variance-tail risk: the possibility that future variance itself becomes extremely high. This is the risk more directly priced by VIX options, and it is naturally represented by the grey factor $\sqrt{Y_\beta}Z_t$. The one-factor gBergomi model forces these two risks to be controlled by the same loading. Empirically, this appears too restrictive. The SPX calibration drives the model toward the rough Bergomi limit, $\beta\to1$, because it wants a clean leverage-skew mechanism. The VIX calibration drives the model toward small $\beta$, because it wants heavy-tailed variance convexity. A two-loading model keeps the same common variance process but allows crash skew and variance-tail convexity to be priced separately. The proposed extension remains economically interpretable. One loading measures how violently volatility reacts to equity selloffs, while the other measures how much independent volatility-spike convexity is priced in the variance market. This provides a more transparent alternative to highly flexible multi-factor specifications whose parameters are difficult to interpret economically. We leave the development, analysis, and calibration of such multi-factor model for our follow-up paper.}

\blue{
\begin{remark}
Given ultra-short-dated SPX and VIX quotes, one can invert the ATM level, skew and curvature expansions of Sections~\ref{sec:VIX_asymptotics}--\ref{sec:SPX_asymptotics} to obtain a suitable initial guess for the VIX–SPX joint calibration. The main caveat is that ultra-short-dated liquidity is asymmetric: SPX options trade with daily expiries, whereas the shortest reliable VIX maturities are weekly, so the VIX side of the initialisation must use the nearest available short slice and a smoothed, arbitrage-free interpolation rather than raw quotes.
\end{remark}
}

\blue{
\begin{remark}[Comparison with rough Bergomi and rough Heston]\label{rem:Comparison}
It is instructive to position the gBergomi model against the two most closely related
rough volatility specifications.\newline 
\emph{Rough Bergomi.} The gBergomi model nests rough Bergomi as boundary case $\beta = 1$. 
The comparison is therefore not
between two distinct calibrations but a sensitivity analysis in the single parameter
$\beta$: holding $(H,\eta,\xi_0)$ at their VIX-calibrated values, the case $\beta = 1$
recovers the rough Bergomi forward-variance dynamics, whose one-factor log-Normal structure is well known to be insufficiently convex to reproduce the market VIX smile jointly with the SPX skew~\cite{Bayer2015PricingVolatility}. This is also evident in Figure~\ref{fig:vixsmilerbergomi} which uses the SPX calibrated values. 
\begin{figure}[ht]
    \centering
    \includegraphics[width=0.75\linewidth]{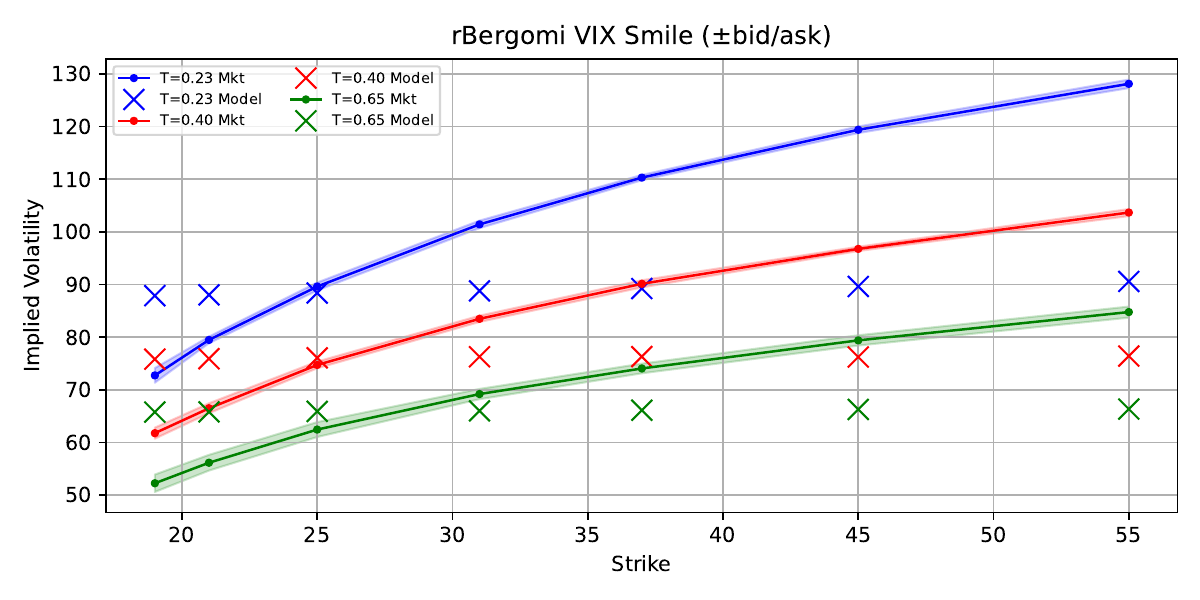}
    \caption{Rough Bergomi VIX smile on 25/10/23.
    $(H, \beta, \eta) = (0.065, 1, 1.09)$.}
    \label{fig:vixsmilerbergomi}
\end{figure}
Reducing $\beta$ below one introduces a heavy-tailed M-Wright mixture of log-Normals, and it is precisely this mechanism that generates the additional VIX-smile convexity. In this sense $\beta$ does not merely add a degree of freedom: it activates a variance-tail channel that is structurally absent from rough Bergomi, which is the empirical flexibility the extension is designed to provide. \newline
\emph{Rough Heston.} Rough Heston instead retains an affine square-root variance process, which confers analytical tractability that gBergomi does not share: a semi-explicit characteristic function through the fractional Riccati equation, hence Fourier-based pricing, together with well-developed finite-dimensional Markovian
lifts~\cite{abi2019multifactor}. Grey Bergomi has neither a transform nor an exact Markovian representation, and pricing relies on the Monte Carlo schemes of
Section~\ref{sec:algos} (though a Markovian approximation is discussed in Appendix~\ref{apx:markovian_approximation}). This tractability, however, comes with a lighter variance tail: in one-factor form rough Heston exhibits the same difficulty in reproducing the joint SPX/VIX surface, which has motivated multi-factor and polynomial extensions such as the quintic Ornstein--Uhlenbeck model~\cite{jaber2022quintic}. The resulting trade-off is clear. Rough Heston preserves transform-based tractability and Markovian lifts but shares the one-factor joint-calibration limitation; gBergomi forgoes that tractability in exchange for a heavy-tailed, M-Wright-driven variance channel that directly targets VIX convexity. This is consistent with our calibration findings above and with the two-loading extension we propose, which retains the grey variance channel while restoring an independent leverage-skew mechanism.
\end{remark}
}

%%%%%%%%%%%%%%%%%%%%%%%%%%%%%%%%%%%%%%%%%%%%
%%%%%%%%%%%%%%%%%%%%%%%%%%%%%%%%%%%%%%%%%%%%
\newpage
\bibliography{bibliography}
\bibliographystyle{siam}
%%%%%%%%%%%%%%%%%%%%%%%%%%%%%%%%%%%%%%%%%%%%
%%%%%%%%%%%%%%%%%%%%%%%%%%%%%%%%%%%%%%%%%%%%

\appendix
\section{M-Wright Variate Generation}\label{apx:mwrightgen}
\blue{
We recall~\eqref{eq:LaplaceM-Wright}, showing that the Laplace transform of the M-Wright density reads
$$
\phi_Z(\lambda) = \int_0^{\infty} e^{-\lambda z} \mathcal{M}_{\beta} (-z) dz = \Ef_{\beta}(-\lambda).
$$
Following~\cite{dextermwright}, let $N_{\beta}(t)$ be the number of events up to $t>0$ in the $\beta$-fractional Poisson ($\beta$-fP) process \cite{Repin2000} with intensity $\eta > 0$ and consider the scaled $\beta$-fP random variable 
$$
Z_{\beta} := \frac{N_{\beta}(t)}{\eta t^{\beta}}.
$$
It can be shown that $\phi_{Z_{\beta}}(\lambda) \to \phi_{Z}(\lambda)$ as $t\uparrow\infty$, where $\phi_{Z_{\beta}}$ is the Laplace transform of the density of $Z_{\beta}$ \cite{dextermwright}. 
By uniqueness of the Laplace transform, the limiting distribution of the scaled $\beta$-fP random variable~$Z_{\beta}$ converges to the M-Wright function as $t\uparrow\infty$. Furthermore, one can write the limiting distribution of $Z_{\beta}$ and $\mathcal{M}_{\beta}$ as a one-sided $\beta$-stable density $s^{(\beta)}(\xi)$ with Laplace transform $\exp(-\lambda^{\beta})$ as
$$
\mathcal{M}_{\beta}(-z) = f_{\beta}(z) = \frac{z^{-1-\frac{1}{\beta}}}{\beta}s^{(\beta)}(z^{-\frac{1}{\beta}}).
$$
By a change of variables $r = z^{-\frac{1}{\beta}}$, one obtains
$$
\frac{\beta}{r^{\beta + 1}} \mathcal{M}_{\beta}(-r^{-\beta}) = \frac{\beta}{r^{\beta + 1}} f_{\beta}(r^{-\beta}) = s^{(\beta)}(r),
$$
so that
$Z \overset{(d)}{=} R^{-\beta}$,
where~$Z$ is a one-sided M-Wright random variable and $R$ a random variable from the one-sided $\beta$-stable distribution $s^{(\beta)}(r)$. Translating the implementation of this in R \cite{rmwright} to Python allows for the variate generation in Figure~\ref{fig:mwright}. 
}
\begin{figure}[h]
    \centering
    {\includegraphics[width=1.0\linewidth]{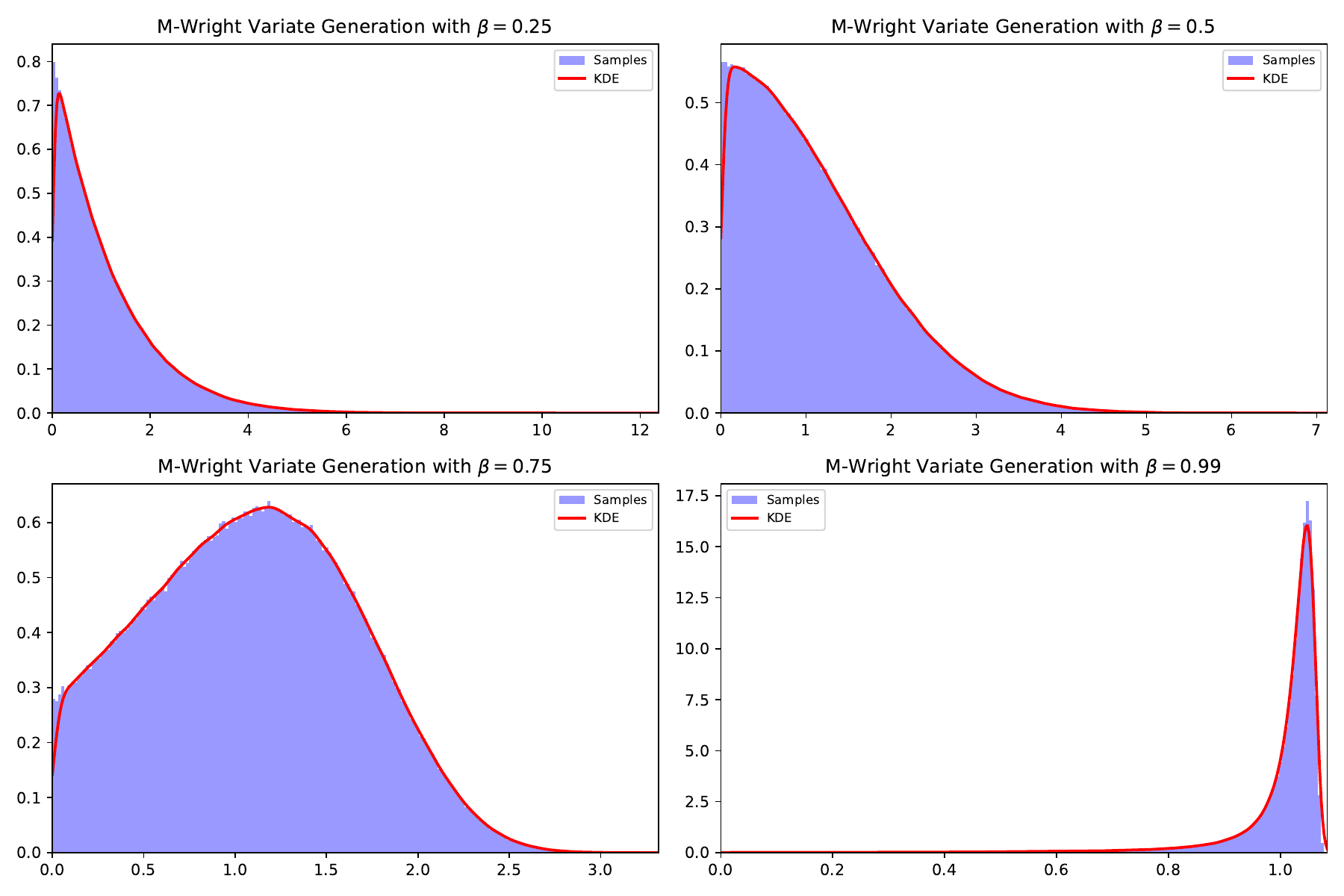}}
    \caption{M-Wright density and $10^6$ M-Wright variates for~$\beta \in \{0.25, 0.5, 0.75, 0.99\}$.}
    \label{fig:mwright}
\end{figure}

%%%%%%%%%%%%%%%%%%%%%%%%%%%%%%%%%%%%%%%%%%%%
%%%%%%%%%%%%%%%%%%%%%%%%%%%%%%%%%%%%%%%%%%%%

\section{Markovian Approximation of grey Bergomi}\label{apx:markovian_approximation}
Since the variance process in gBergomi is not Markovian, simulation schemes are computationally more expensive. A workaround is to consider Markovian approximations~\cite{abi2019multifactor, bayer2022markovian, Carmona2000}. 
Writing the power law kernel $K(t) := \frac{1}{\Gamma(\Hp)}t^{\Hm}$ as the Laplace transform
$$
K(t) = \int_0^{\infty} \E^{-tx} \mu(\D x), 
\quad \text{with } \mu(\D x) := \frac{\D x}{\Gamma(\Hm)\Gamma(\Hp) x^{\Hp}},
$$
an application of the stochastic Fubini theorem yields
$$
\int_0^t K(t-s) \D B_s = \int_0^{\infty} \int_0^t \E^{-(t-s)x} \D B_s \mu (\D x)  = \int_0^{\infty} Y_t^x \mu(\D x),
$$
where $Y_t^x := \int_0^t \E^{-(t-s)x} \D B_s$ is an Ornstein-Uhlenbeck process.
Approximating the measure~$\mu$ by a finite sum of Dirac masses yields 
$$
K(t) = \int_0^{\infty} \E^{-tx} \mu(\D x) \approx \sum_{i=0}^N w^N_i \E^{-x^N_i t},
$$
where $(w^N_i)_{i=1}^N$ are positive weights and $(x^N_i)_{i=1}^N$ are mean-reverting speeds. Applying this to the gBergomi model results in the Markovian approximations $S^N$ and $V^N$,
\begin{align}
    &S^N_t = S^N_0 \exp\left\{- \half \int_0^t V^N_s \D s + \int_0^t \sqrt{V^N_s} \left(\rho \D B_s + \sqrt{1-\rho^2}\D B^{\bot}_s\right)\right\},\\
    &V_t^N = \frac{\xi_0(t)}{\Ef_{\beta}(\cff t^{2H})} \exp\left\{\eta \sqrt{Y_\beta} \int_0^t \sum_{i=0}^N w^N_i \E^{-x^N_i (t-s)} \D B_s\right\},
\end{align}
where $S^N_0 = S_0 = 1$ and \blue{$V_0^{N} = V_0 = \xi_0(0).$} Given the positive weights $(w^N_i)_{i=1}^N$ and the mean-reverting speeds $(x^N_i)_{i=1}^N$, one can simulate the Markovian approximation of the gBergomi model. 
Carefully choosing weights and mean-reverting speeds can drastically improve the value of~$N$ one needs in order to obtain a given accuracy. An example of this can be seen in comparing the methodology in~\cite{bayer2022markovian} and~\cite{Zhu_2021}, where in~\cite{bayer2022markovian} one can afford to use a considerably smaller $N$, thus, decreasing simulation times.

%%%%%%%%%%%%%%%%%%%%%%%%%%%%%%%%%%%%%%%%%%%%
%%%%%%%%%%%%%%%%%%%%%%%%%%%%%%%%%%%%%%%%%%%%

\section{Empirical analysis of the VIX distribution}\label{apx:VIX_smile_rBergomi}
We provide here an empirical analysis of the~VIX distribution, discussing the log-Normal assumption for the~VIX. In the rough Bergomi model~\eqref{eq:rBergomi_dynamics} with $\Gm = B^{H}$,
the volatility exhibits log-Normal dynamics,
so that the VIX, as an integral over a short one-month interval of log-Normal variables, 
is close to log-Normal by~\cite{Fenton1960TheSystems}. 
Figure~\ref{fig:VIX_dist} shows the distribution of the  VIX \blue{(daily)} log-returns during 3/1/05--28/10/22\footnote{VIX historical data obtained from the CBOE (\url{https://www.cboe.com}) ($4491$ daily observations).};
clearly, the presence of residual skewness and kurtosis cannot be fully accounted by a Gaussian assumption.
We further compare the VIX log-returns using QQ plots in Figure~\ref{fig:QQplots}
to various probability distributions over different periods. 
While the Normal distribution fails to capture the tail behaviour of VIX returns, both the Student's t and Laplace distributions exhibit improved fit in the tails. 
A notable limitation is the inability of these distributions to capture the asymmetry present in the VIX returns fully, 
although the severity of the misfit differs through the time periods.
Table~\ref{tab:normality_test} summarises the Gaussian test hypothesis on the VIX log-returns. 
The very low D'Agostino $p$-value~\cite{DAgostino1990ANormality}
rejects the null Gaussian hypothesis based on the skewness values,
confirmed by the small $p$-value 
from the Anscombe's test~\cite{Anscombe1983DistributionSamples},
based on the kurtosis,
and by the norm test statistic~\cite{DAgostino1973Tests1, DAgostino1971AnSamples}.
Finally, the Augmented Dickey-Fuller on the VIX log-returns over the whole period indicates stationarity.

\begin{figure}[ht]
    \centering
    \includegraphics[scale=0.33]{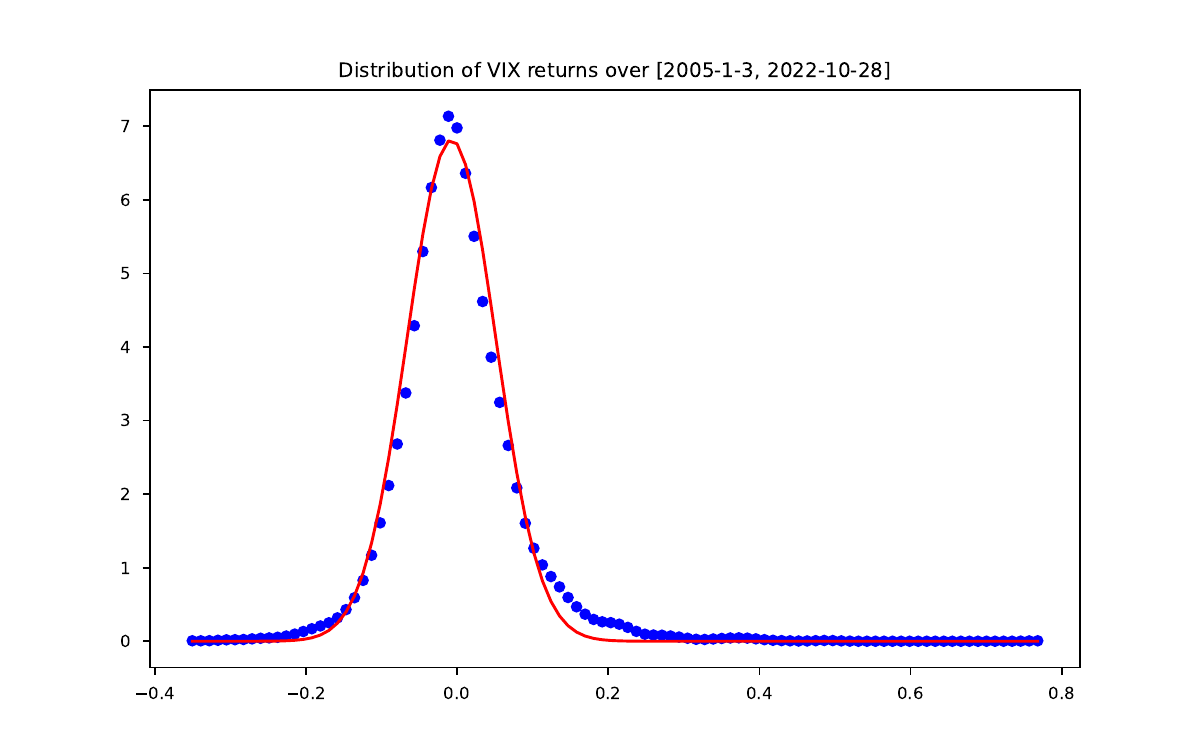}
    \centering
    \begin{subfigure}[b]{\textwidth}
        \centering
        \includegraphics[scale=0.33]{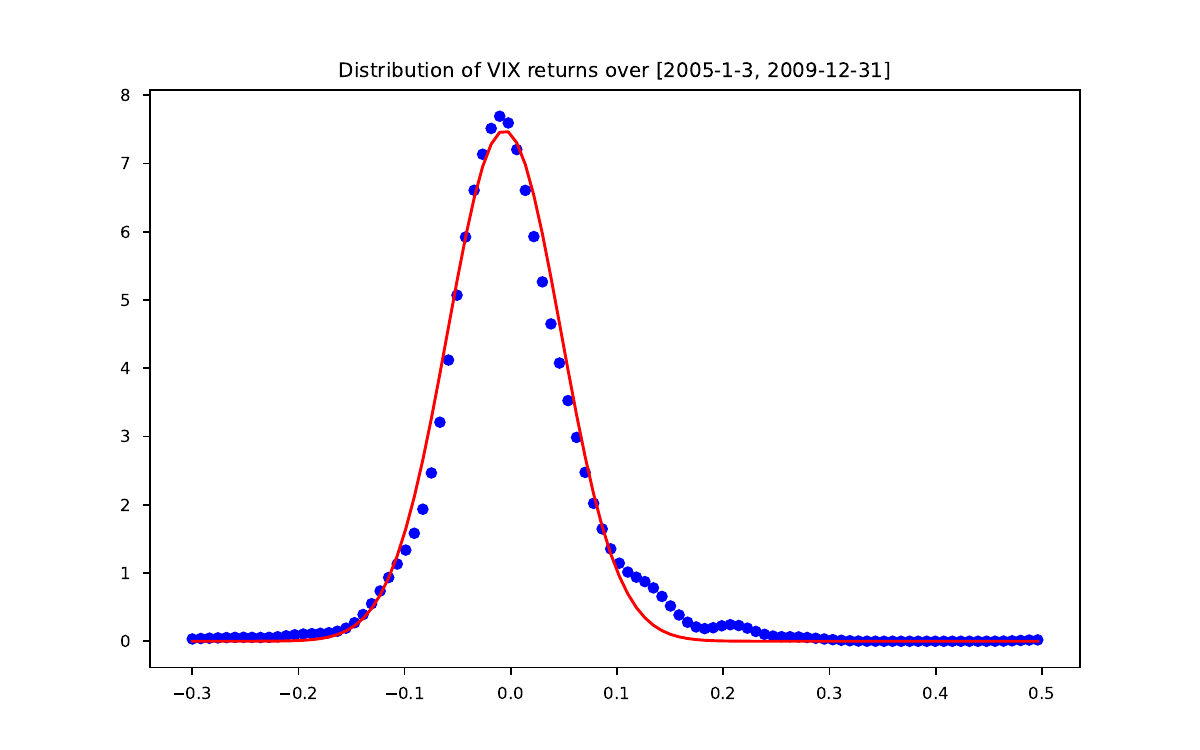}
        \includegraphics[scale=0.33]{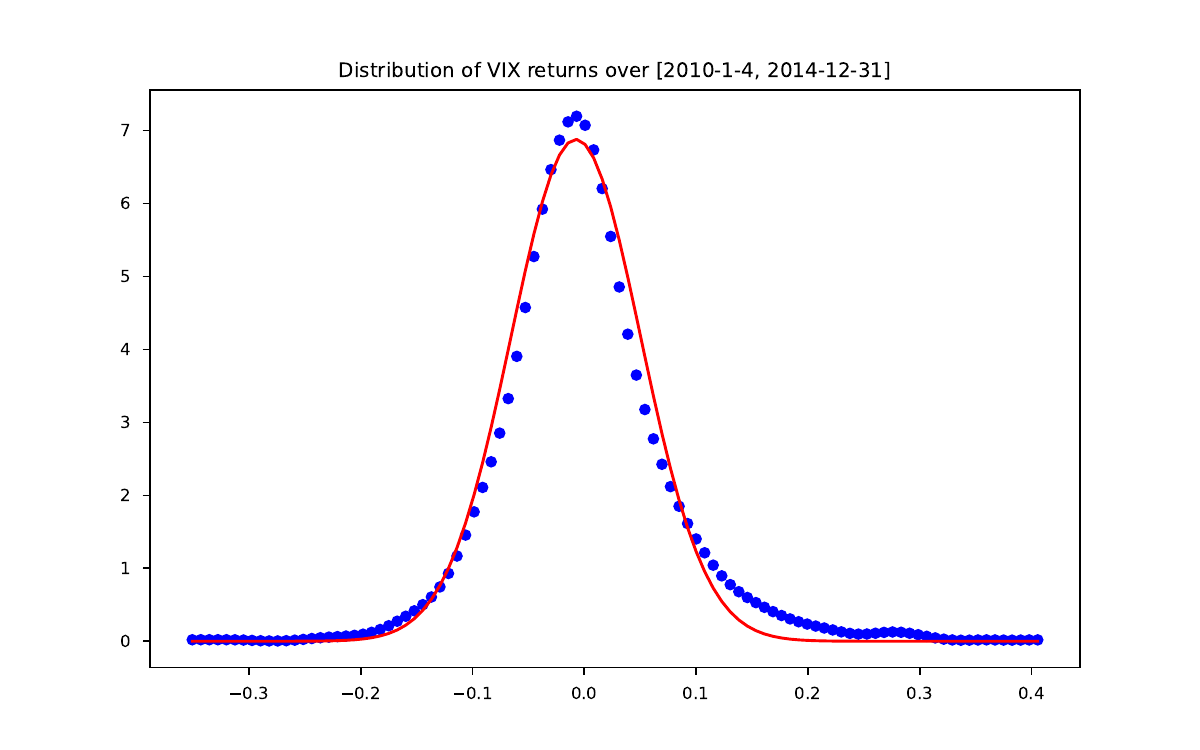}
        \includegraphics[scale=0.33]{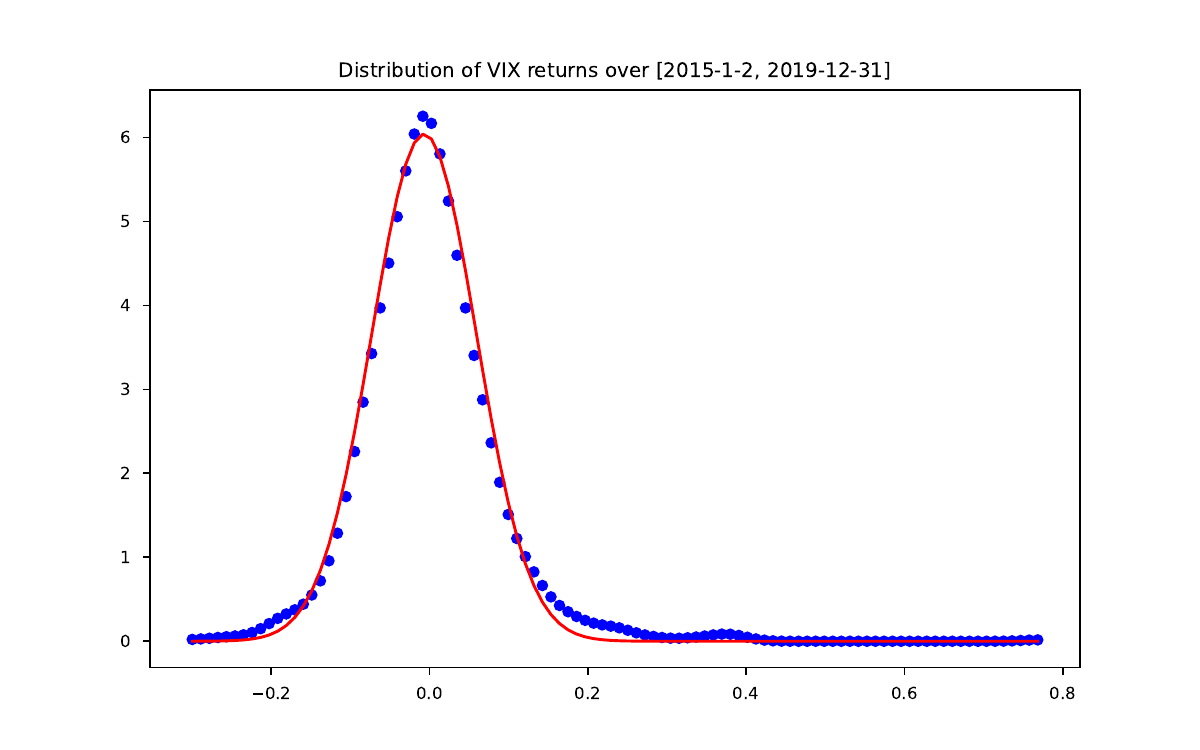}
        \includegraphics[scale=0.33]{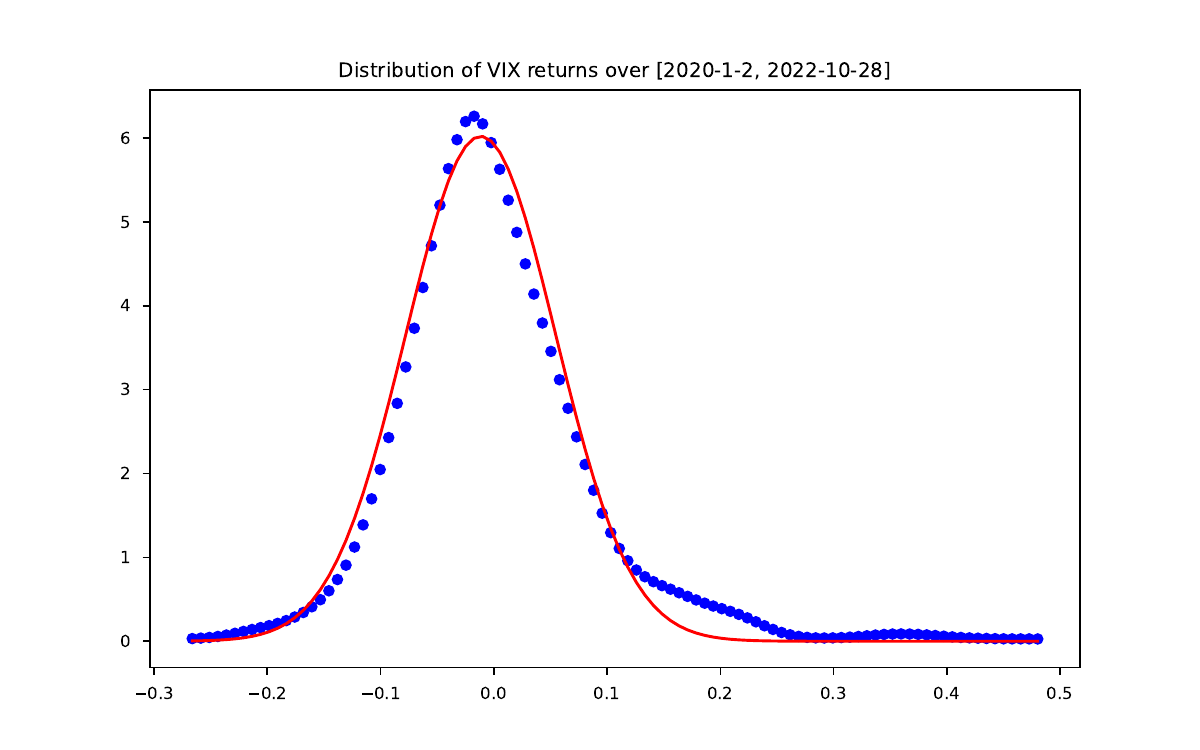}
    \end{subfigure}
    \caption{Gaussian fit to VIX log-returns over several time periods.}
    \label{fig:VIX_dist}
\end{figure}

\begin{figure}[htbp]
    \centering
    \begin{subfigure}[b]{0.32\textwidth}
        \centering
        \includegraphics[width=\textwidth]{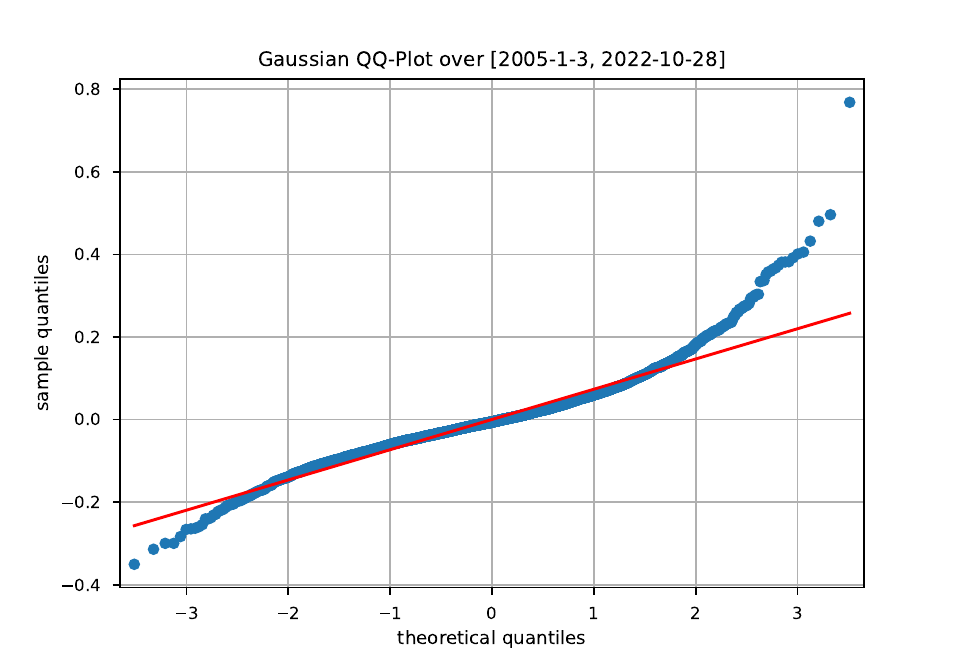}
        \caption{Gaussian}
        % \label{fig:image1}
    \end{subfigure}
    \hfill
    \begin{subfigure}[b]{0.32\textwidth}
        \centering
        \includegraphics[width=\textwidth]{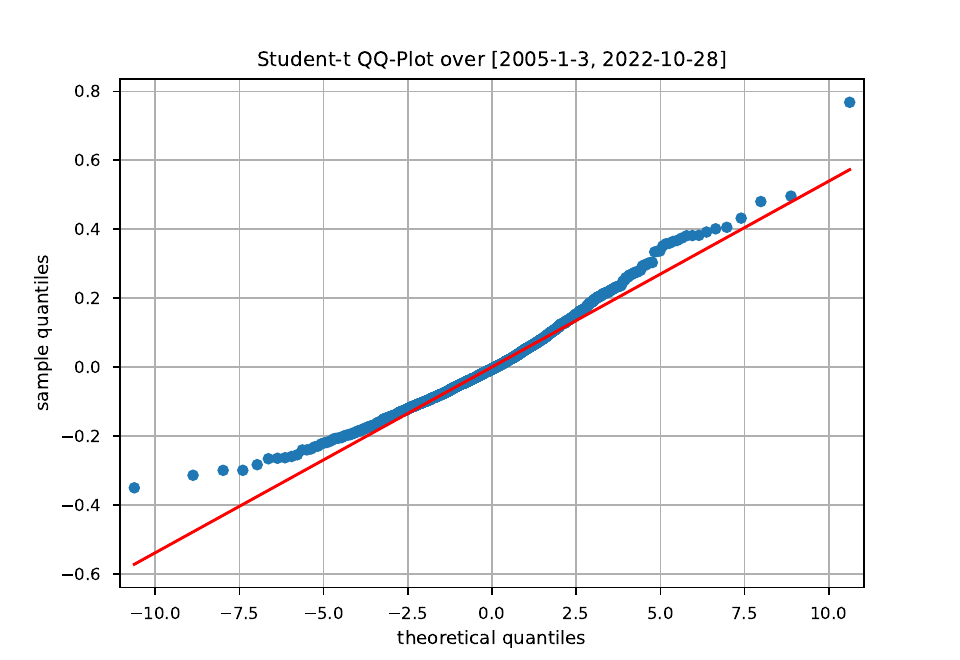}
        \caption{Student's t}
        % \label{fig:image2}
    \end{subfigure}
    \hfill
    \begin{subfigure}[b]{0.32\textwidth}
        \centering
        \includegraphics[width=\textwidth]{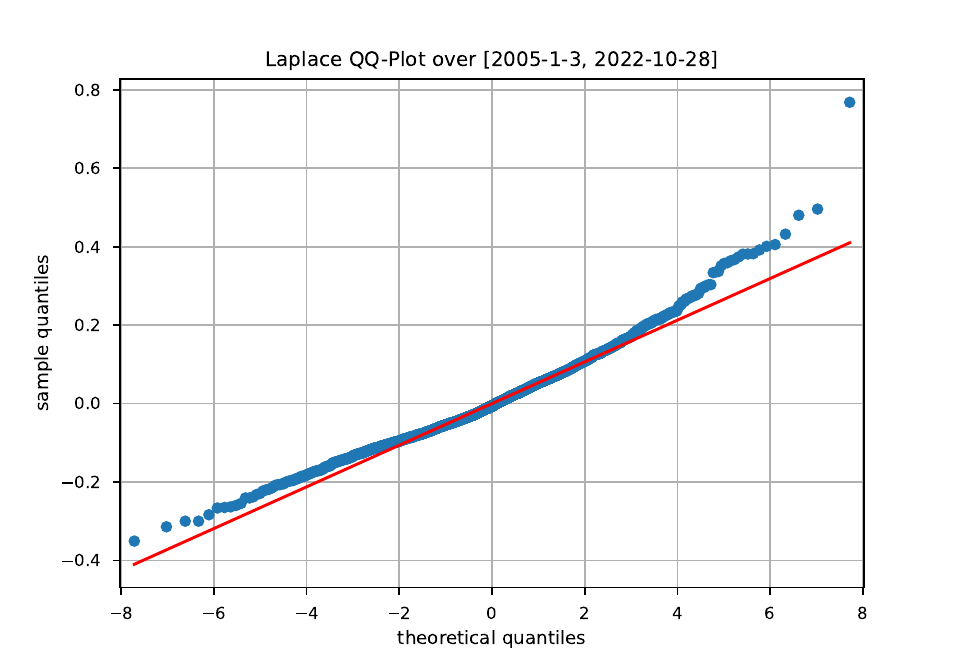}
        \caption{Laplace}
        % \label{fig:image3}
    \end{subfigure}

    \centering
    \begin{subfigure}[b]{0.32\textwidth}
        \centering
        \includegraphics[width=\textwidth]{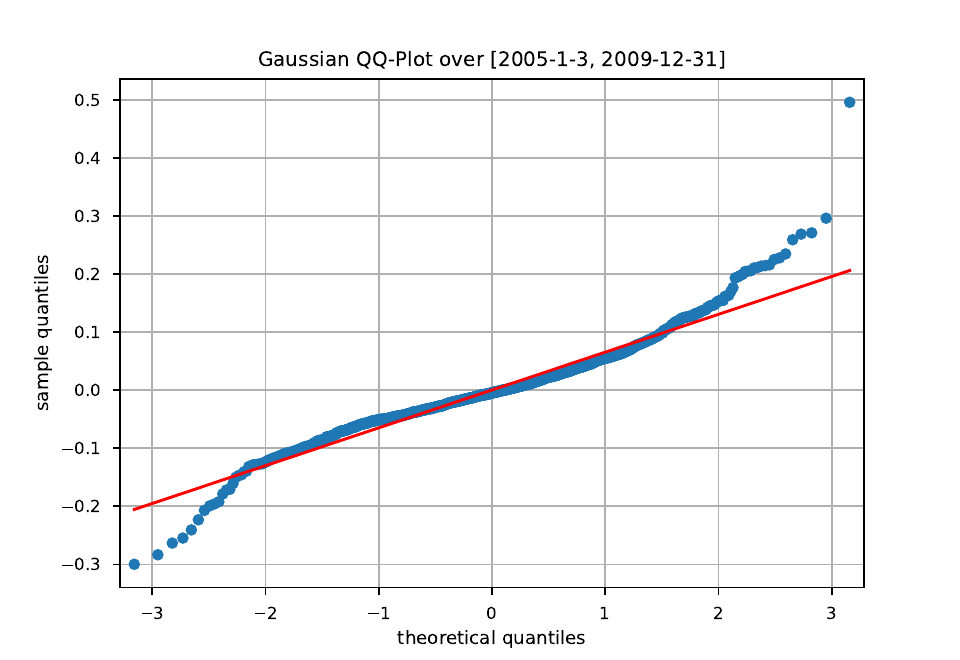}
        \caption{Gaussian}
    \end{subfigure}
    \hfill
    \begin{subfigure}[b]{0.32\textwidth}
        \centering
        \includegraphics[width=\textwidth]{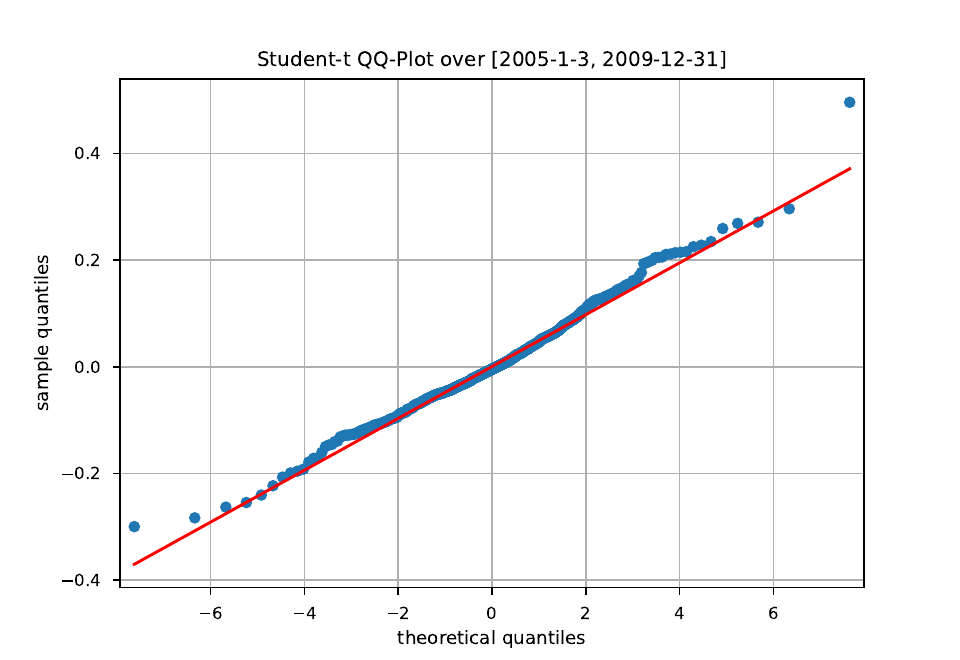}
        \caption{Student's t}
    \end{subfigure}
    \hfill
    \begin{subfigure}[b]{0.32\textwidth}
        \centering
        \includegraphics[width=\textwidth]{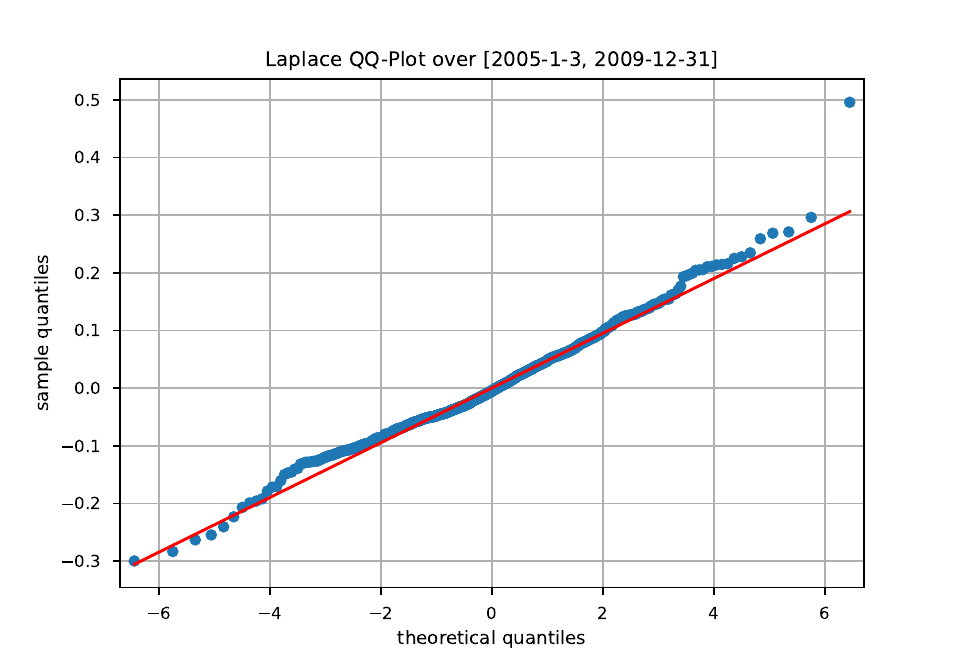}
        \caption{Laplace}
    \end{subfigure}

    \centering
    \begin{subfigure}[b]{0.32\textwidth}
        \centering
        \includegraphics[width=\textwidth]{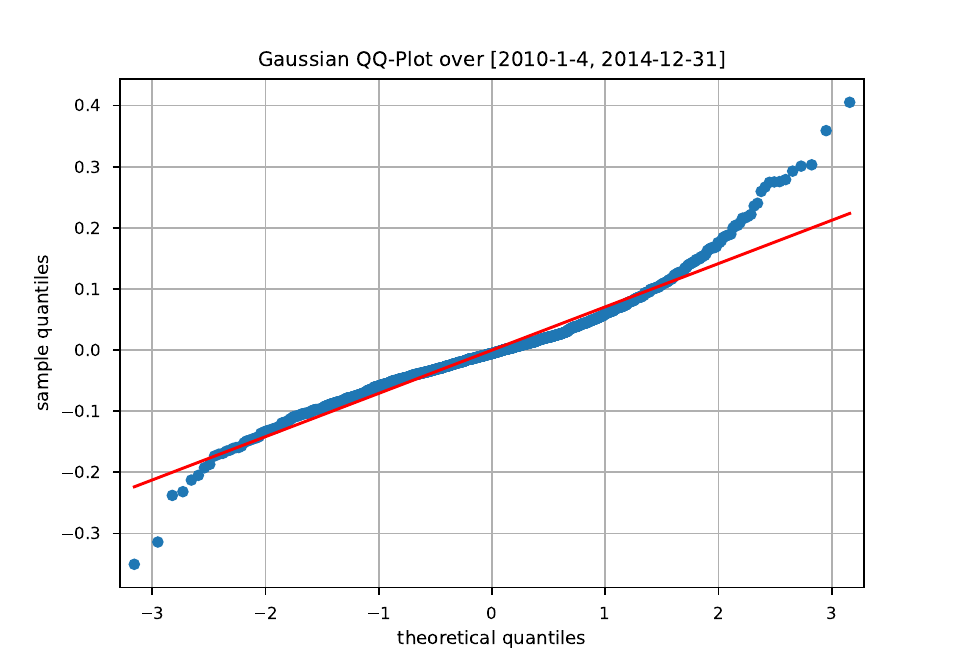}
        \caption{Gaussian}
    \end{subfigure}
    \hfill
    \begin{subfigure}[b]{0.32\textwidth}
        \centering
        \includegraphics[width=\textwidth]{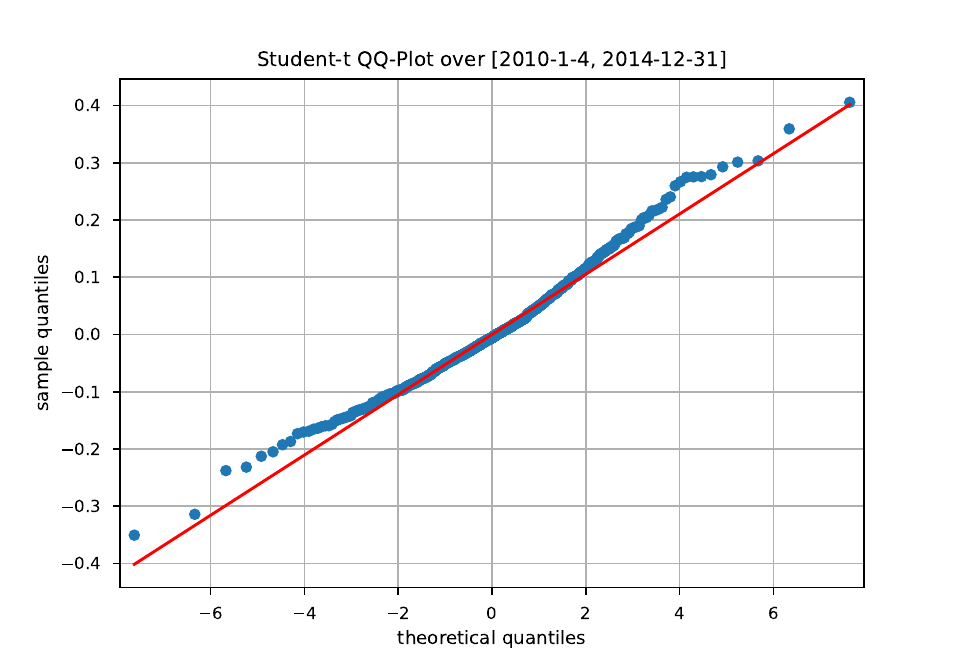}
        \caption{Student's t}
    \end{subfigure}
    \hfill
    \begin{subfigure}[b]{0.32\textwidth}
        \centering
        \includegraphics[width=\textwidth]{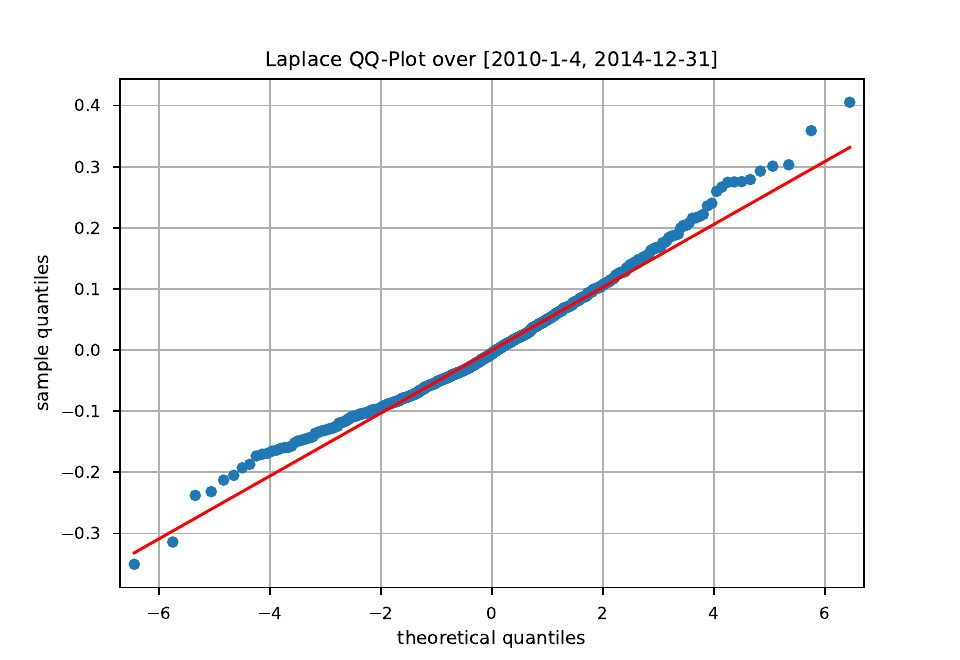}
        \caption{Laplace}
    \end{subfigure}

    \centering
    \begin{subfigure}[b]{0.32\textwidth}
        \centering
        \includegraphics[width=\textwidth]{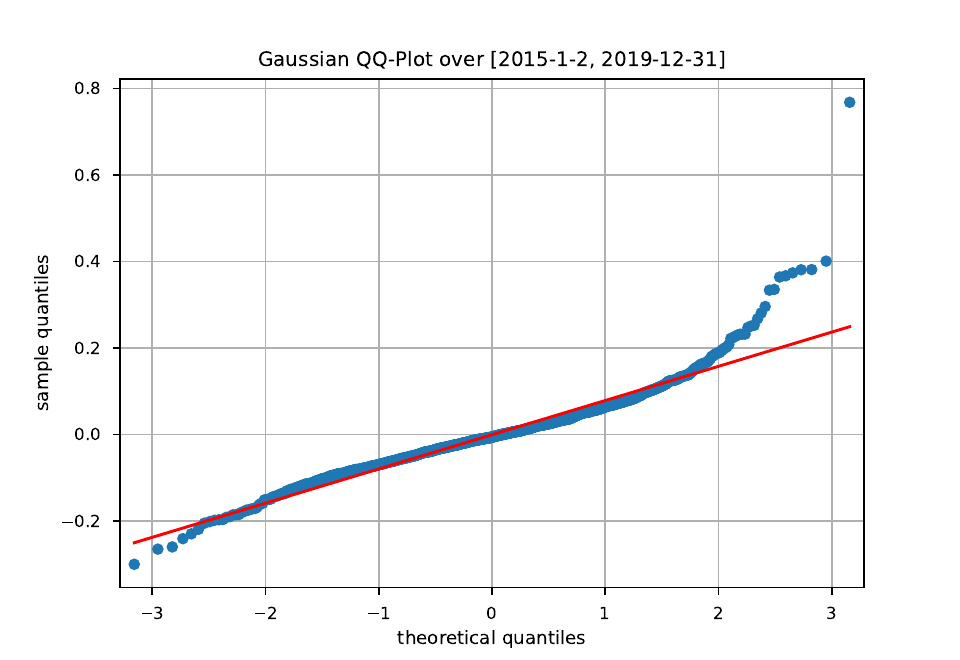}
        \caption{Gaussian}
    \end{subfigure}
    \hfill
    \begin{subfigure}[b]{0.32\textwidth}
        \centering
        \includegraphics[width=\textwidth]{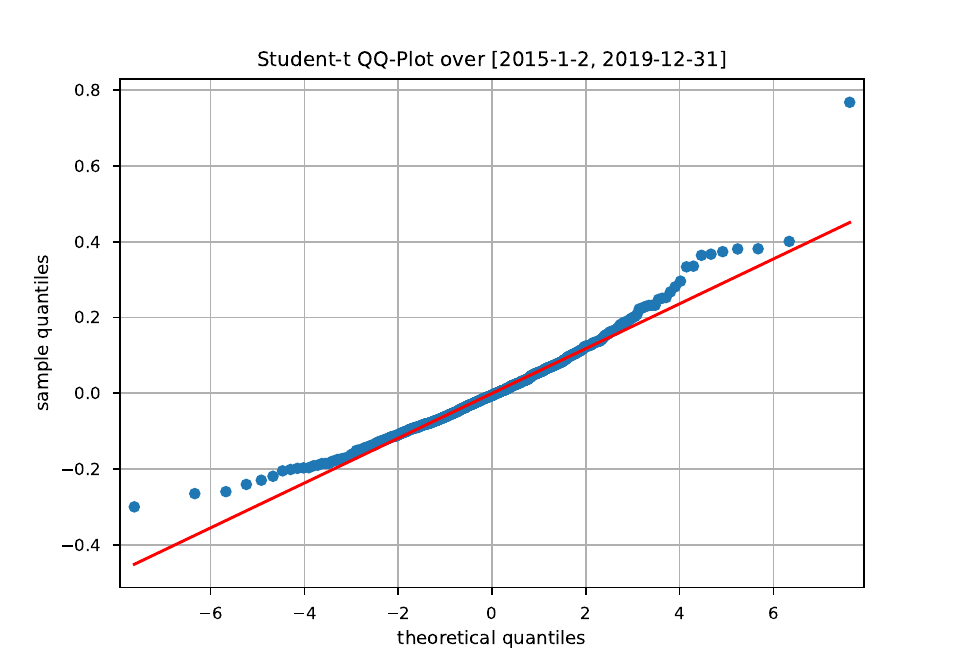}
        \caption{Student's t}
    \end{subfigure}
    \hfill
    \begin{subfigure}[b]{0.32\textwidth}
        \centering
        \includegraphics[width=\textwidth]{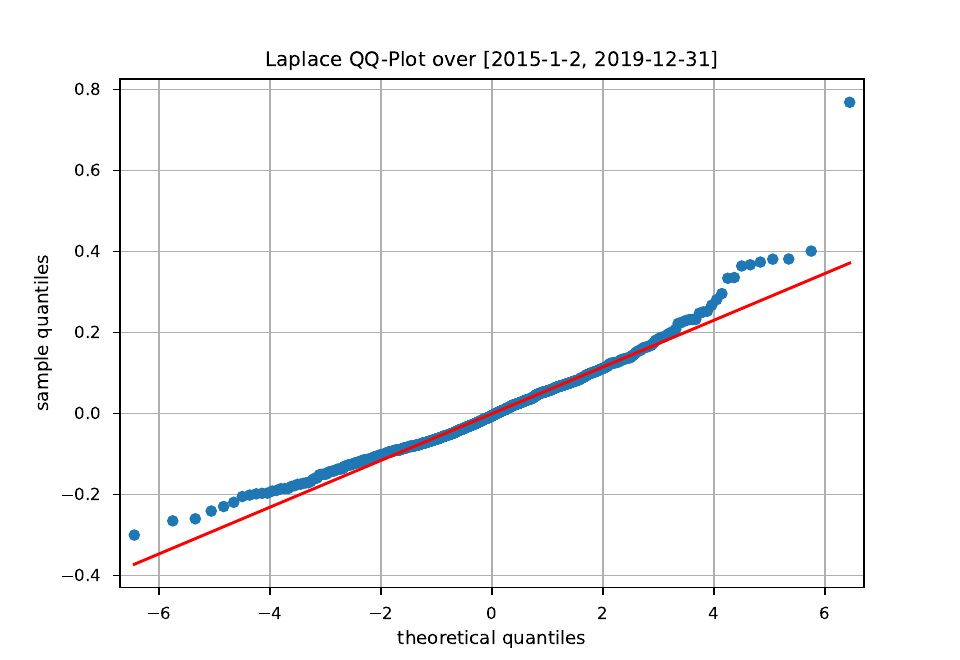}
        \caption{Laplace}
    \end{subfigure}

    \centering
    \begin{subfigure}[b]{0.32\textwidth}
        \centering
        \includegraphics[width=\textwidth]{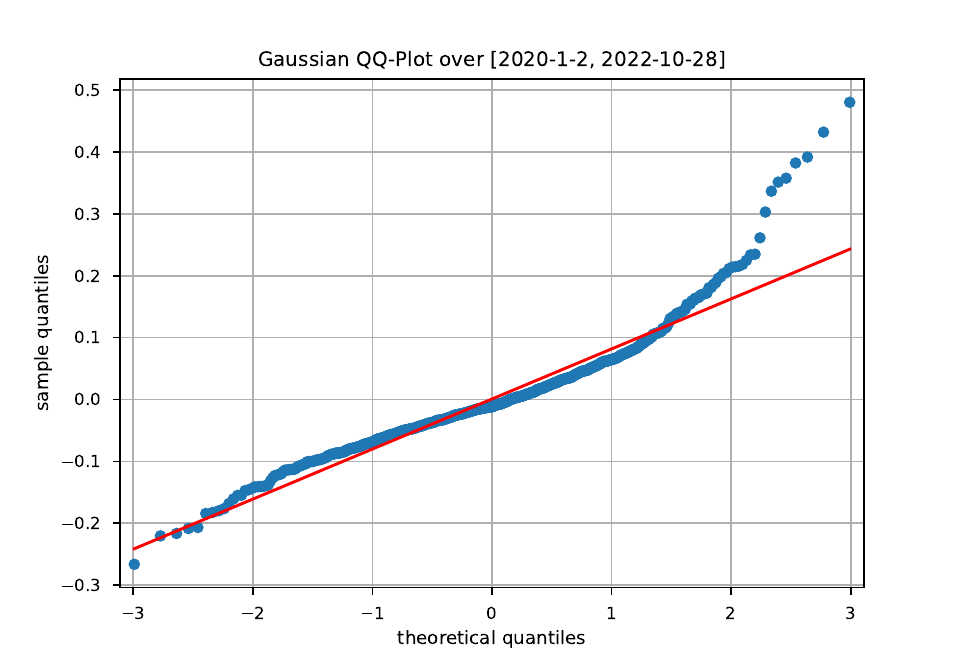}
        \caption{Gaussian}
    \end{subfigure}
    \hfill
    \begin{subfigure}[b]{0.32\textwidth}
        \centering
        \includegraphics[width=\textwidth]{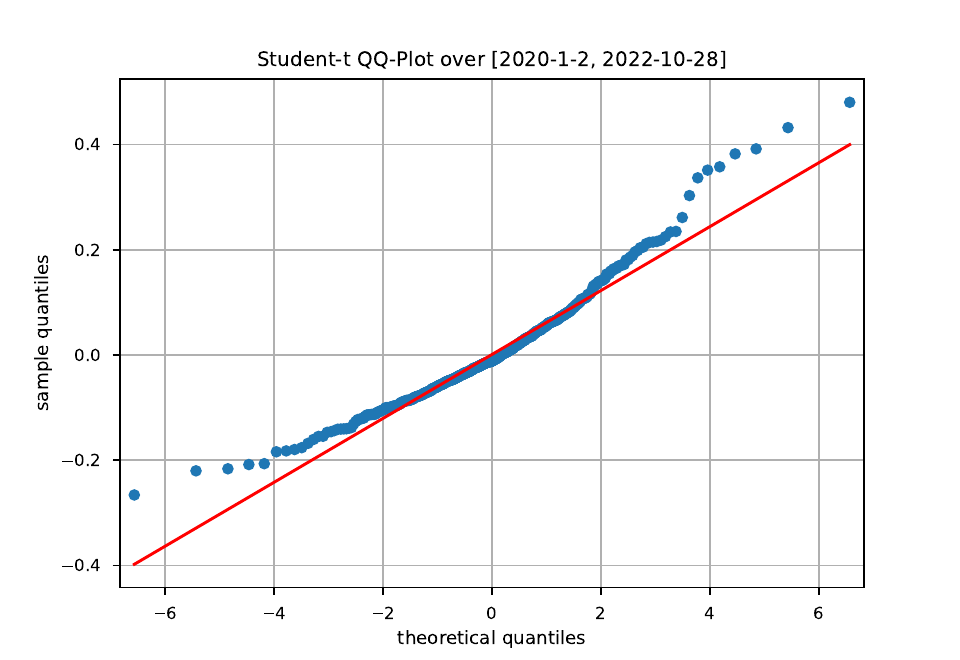}
        \caption{Student's t}
    \end{subfigure}
    \hfill
    \begin{subfigure}[b]{0.32\textwidth}
        \centering
        \includegraphics[width=\textwidth]{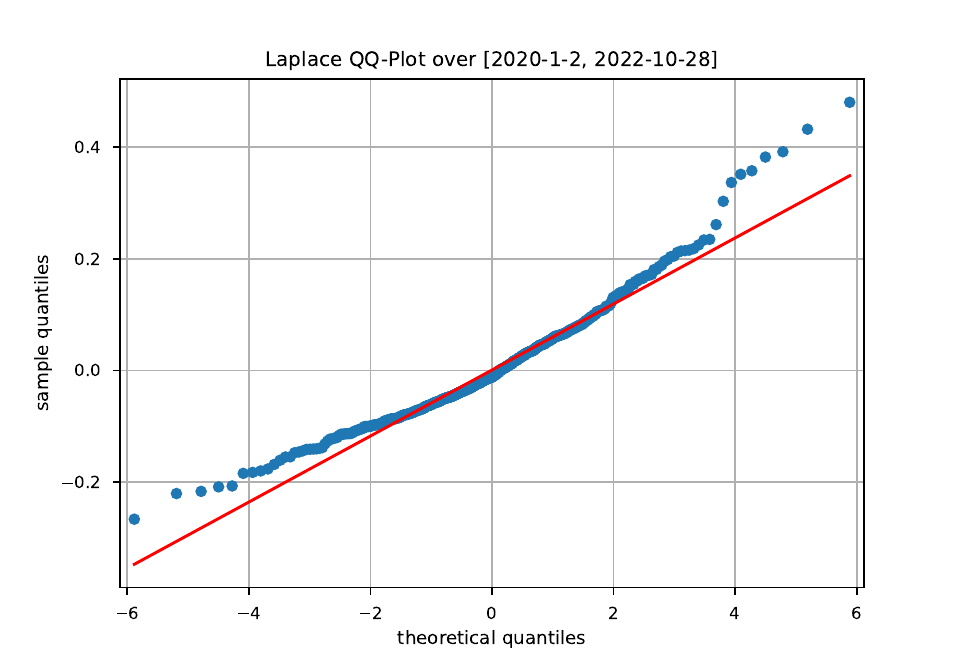}
        \caption{Laplace}
    \end{subfigure}
    \caption{QQ plots of VIX log-returns against standardised reference distributions over several time periods. The Gaussian and Laplace references have zero location and unit scale, while the Student's~$t$ reference has $\nu=4$ degrees of freedom.}\label{fig:QQplots}
\end{figure}

\begin{table}[hbt!]
\begin{tabular}{@{}cccccc@{}}
\toprule
\textbf{Period}    & \textbf{Skewness} & \textbf{Skewness $p$-value} & \textbf{Kurtosis} & \textbf{Kurtosis $p$-value} & \textbf{Norm test} \\ \midrule
{2005-2022} & 1.05              & 2.2e-126                 & 6.10              & 6.7e-128                 & 1.3e-250          \\
{2005-2010} & 0.64              & 1.27e-17                  & 4.57              & 2.19e-31                  & 4.42e-46           \\
{2010-2015} & 0.70              & 3.32e-20                  & 3.37              & 1.49e-24                  & 7.47e-42           \\
{2015-2020} & 1.34              & 1.12e-51                  & 8.59              & 7.31e-48                  & 2.84e-96           \\
{2020-2022} & 1.28              & 8.08e-29                  & 4.62              & 1.93e-19                  & 2.51e-45 \\ \bottomrule
\end{tabular}
\caption{Normality tests on VIX log returns over different periods. The tests evaluate the skewness~\cite{DAgostino1990ANormality} and kurtosis~\cite{Anscombe1983DistributionSamples} of the underlying population compared to a Normal distribution, as well as the overall adherence of the sample to a Normal distribution~\cite{DAgostino1973Tests1, DAgostino1971AnSamples}. 
}\label{tab:normality_test}
\end{table}

%%%%%%%%%%%%%%%%%%%%%%%%%%%%%%%%%%%%%%%%%%%%

\end{document}